\documentclass[article]{amsart}
\date{\today}

\usepackage{amsmath,amsthm}
\usepackage{amssymb}
\usepackage{bbm}

\newtheorem{proposition}{Proposition}
\newtheorem{lemma}{Lemma}

\newtheorem{theorem}{Theorem}

\newtheorem{remark}{Remark}

\newcommand{\N}{\mathbb{N}} 
\newcommand{\C}{\mathbb{C}} 
\newcommand{\Z}{\mathbb{Z}} 
\newcommand{\R}{\mathbb{R}} 

\newcommand{\bp}{{\bf p}}
\newcommand{\bx}{{\bf x}}

\newcommand{\bu}{{\bf u}}
\newcommand{\bs}{\tilde{B}}
\newcommand{\DA}{D_{{\bf A}}}
\newcommand{\SP}[2]{\big\langle #1,#2 \big\rangle} 
\newcommand{\sps}[2]{\langle #1,#2 \rangle} 

\newcommand{\ri}{\mathrm{i\,}}
\newcommand{\rd}{d}

\title[Localization of massless Dirac fermions]{Localization of 
two-dimensional massless Dirac fermions in a magnetic quantum dot}
\author{Martin K\"onenberg}
\address{Martin K\"onenberg\\
Fakult\"at f\"ur Mathematik und Informatik\\
FernUniversit\"at Hagen\\
L\"utzowstra{\ss}e 125\\
D-58084 Hagen, Germany.
{\it Present address:} Fakult\"at f\"ur Physik \\
Universit\"at Wien, Boltzmanngasse 5, 1090 Vienna, Austria.}
\email{martin.koenenberg@univie.ac.at}
\author{Edgardo Stockmeyer}
\address{Edgardo Stockmeyer\\
Mathematisches Institut\\
Ludwig-Maximilians-Universit\"at\\
Theresienstra{\ss}e 39\\
D-80333 M\"unchen, Germany.}
\email{stock@math.lmu.de}

\subjclass{Primary 81Q10; Secondary 47B25}
\keywords{Magnetic operator,
localization,
Dirac operator}
\begin{document}
\maketitle
\begin{abstract}
We consider a two-dimensional massless Dirac operator $H$ in the presence of a perturbed homogeneous magnetic field $B=B_0+b$ and a scalar electric potential $V$. For $V\in L_{\rm loc}^p(\R^2)$, $p\in(2,\infty]$, and  $b\in L_{\rm loc}^q(\R^2)$, $q\in(1,\infty]$, both decaying at infinity, we show that
states in the discrete spectrum of $H$ are superexponentially localized.
We establish the existence of such  states between the zeroth and the first Landau level assuming that $V=0$. In addition, under the condition that $b$ is rotationally symmetric and that $V$ satisfies certain analyticity condition on the angular variable, we show that states belonging to the discrete spectrum of $H$ are Gaussian-like localized.  
\end{abstract}
\section{Introduction}
Graphene is a two-dimensional lattice of carbon atoms arranged on a
honeycomb structure. Due to its unusual properties it has attracted a
great deal of attention since its discovery
\cite{castro2009electronic,novoselov2004electric}.  One of the
striking facts about graphene is that the dynamics of its low-energy
excitations (the charge carriers) can be described by massless
two-dimensional Dirac operators. An interesting feature of Dirac
fermions is the lack of localization under the influence of an
external electric potential \cite{vogelsang,Kalfetal2003}. This fact, related to Klein's paradox \cite{castro2009electronic}, is
due to the peculiar cone-like gapless structure of the spectrum of  massless
free Dirac operators. 

It was suggested in \cite{de2007magnetic} that it is possible to
confine such Dirac fermions in graphene by inhomogeneous magnetic
fields of the type $B=B_0+b$, where $B_0>0$ is a constant and $b$ a
perturbation with negative flux  that decays at infinity. The spectrum of the
corresponding Dirac operator in a constant magnetic field $B_0$ is given by
the (relativistic) Landau levels. The idea is that as the perturbation
$b$ is turned on eigenvalues will emerge from the Landau levels giving
rise to states localized on the bulk of the support of $b$. In this
manner a so-called (magnetic) quantum dot or artificial atom can be
created.  These type of models, also with an external electric
potential $V$, have been further studied in the physics literature, for instance,  in
\cite{martino2010spectrum,masir,wang2009magnetically,kormanyos2008bound}
for the one particle case and in \cite{hausler2009artificial,EDSS}
for the multiparticle case.  The articles
\cite{martino2010spectrum,masir,wang2009magnetically,kormanyos2008bound}
deal with specific electromagnetic fields for which the model is partly
solvable or suitable for numerical computations. 

In this article we consider a large class of electromagnetic
perturbations $(b,V)$ with $V\in L_{\rm loc}^p(\R^2)$,
$p\in(2,\infty]$, and $b\in L_{\rm loc}^q(\R^2)$, $q\in(1,\infty]$,
both decaying at infinity. The essential spectrum of the corresponding
Dirac-operator $H$ describing the quantum dot is given by the Landau
levels. We show that eigenfunctions belonging to the discrete spectrum
of $H$ are superexponentially localized, i.e., they decay faster than
any exponential.  In the case when $V=0$ we verify the existence of
eigenvalues between the zeroth and the first Landau-level assuming
that $b< 0$. Assuming that a certain analyticity conditions on the
angular variable of $V$ is fulfilled and that $b$ is rotationally
symmetric we prove that those states are actually Gaussian-like
localized. These type of results on superexponential and Gaussian
localization, although new for Dirac operators, are known to hold for
spinless magnetic Schr\"odinger operators
\cite{CorneanNenciu1998,Erdos,Nakamura1996,Sordoni1998}. We benefit from this
insight to prove our statements. A precise description of our results 
is given in the next section.
\section{Results}
We consider the massless two-dimensional Dirac operator with an external
magnetic field $B:\R^2\to \R$, pointing
perpendicularly to the plane, and an electric potential  $V:\R^2\to \R$. 
We are interested in the  Hamiltonians
\begin{align}
  \label{eq:1}
  D_{\bf A}&:= {\boldsymbol \sigma}\cdot(\bp-{\bf A}),\\
\label{leila}
H&:=D_{\bf A}+V\,,
\end{align}
a priori defined on $C_0^\infty(\R^2;\C^2)\subset L^2(\R^2;\C^2)$.
Here $\bp:=\tfrac{1}{\ri}\nabla$ is the momentum of the  particle and  $ {\boldsymbol \sigma}:= (\sigma_1,
\sigma_2)$ is a  vector whose entries
\begin{equation*}
  \label{eq:2}
 \sigma_1=:
\left(
\begin{array}{cc}
0&1\\
1&0
\end{array}\right),\qquad
\sigma_2=:
\left(
\begin{array}{cc}
0&-\ri\\
\ri&0
\end{array}\right)\,, 
\end{equation*}
are   Pauli matrices. 
The magnetic field $B$ enters in the definitions \eqref{eq:1} and \eqref{leila} by means of the magnetic vector potential ${\bf A}=(A_1,A_2):\R^2\to \R^2$ through the relation
\begin{equation}\label{carla1}
 B=\partial_1 A_2-\partial_2 A_1=:{\rm curl} \,{\bf A}\,,
\end{equation} 
which is understood in the sense of distributions.

Throughout this article we assume the following on $(B,V)$:
 \begin{itemize}
\item[(A1)] $B=B_0+b$ where $B_0>0$ is a number and  $b\in L^{q}_{\rm loc}(\R^2;\R)$ for some $q\in(1,\infty]$ and 
$\lim_{n\to \infty} \| \mathbbm{1}_{\{|{\bf x}|\ge n\}}b\|_\infty=0.$
\item[(A2)]  $V\in L_{loc}^{p}(\R^2;\R)$
for some $p\in (2, \infty]$ and  $
\lim_{n\to\infty}\|\mathbbm{1}_{\{|\bx|\ge n\}}V\|_\infty=0\,.$
\end{itemize}
Here $\mathbbm{1}_{I}(\cdot)$ denotes the  characteristic function on the set $I$.
Assuming that $B$ fulfills (A1) we can always find ${\bf A}\in L^t_{\rm
  loc}(\R^2;\R^2)$ for some $t\in (2,\infty]$ satisfying \eqref{carla1} (see Remark \ref{poisson}). For such
magnetic vector potentials and electric potentials $V$ satisfying (A2)
we know that the operators defined in \eqref{eq:1} and \eqref{leila}
are essentially self-adjoint (see Subsection \ref{def-of-H}). We denote their self-adjoint extensions by the same symbols and their domains by $\mathcal{D}(\DA)$ and $\mathcal{D}(H)$ respectively.

To the homogeneous magnetic  field $B_0$ we associate  the vector potential
\begin{align}\label{a0}
{\bf A}_0:=\frac{B_0}{2} (-x_2,x_1)\,,
\end{align}
satisfying ${\rm curl} \,{\bf A}_0=B_0$. It is well known that the spectrum of $ D_{{\bf A}_0}$ consists of infinitely degenerated eigenvalues $(l_n)_{n\in\Z}$, called Landau levels, given by
$$l_n:={\mathrm{sgn}(n)}\sqrt{2|n|B_0},\qquad n\in \Z\,,$$
where $\mathrm{sgn}(n)=n/|n|$  if $n\not= 0$ and equals one if $n=0$. 

Given a   self-adjoint operator $T$ we write $\sigma_{\rm pp}(T), \sigma_{\rm d}(T)$, and $\sigma_{\rm ess}(T)$ to denote the pure point, discrete, and essential spectra of $T$ respectively.
Our first main result is as follows.
\begin{theorem}\label{Teo1}
Assume that $B$ satisfies (A1) and let ${\bf A}\in L^p_{\rm loc}(\R^2;\R^2),\,p\in(2,\infty],$ with ${\rm curl}\, {\bf A} = B$. Then, the spectrum of $\DA$ is symmetric with respect to zero and 
$$\sigma_{\rm ess}(\DA)= (l_n)_{n\in\Z} \,.$$
Moreover,
\begin{itemize}
\item[(a)] If $b\le 0$ and strictly negative on some open set,  then the discrete spectrum of $\DA$ on $(0,l_1)$ is non-empty, i.e.,  $\sigma_{\rm d}(\DA)\cap (0,l_1)\not=\emptyset$ and 
$${\rm dim(Ran}(\mathbbm{1}_{(0,l_1)}(D_{\mathbf{A}})))=\infty\,.$$
\item[(b)] If $b\ge 0$ then 
$${\rm dim(Ran}(\mathbbm{1}_{(0,l_1)}(D_{\mathbf{A}})))=0\,.$$
\end{itemize} 
\end{theorem}
This theorem is a consequence of lemmas \ref{Cor} and \ref{lemma}. That the spectrum of $\DA$ is symmetric with respect to zero is well known; see, however,  Proposition \ref{supersymmetry}.
\begin{remark}
  A similar result to Theorem \ref{Teo1} is shown in \cite{Besch2000}
  when $b$ is replaced by $\lambda b$ and $\lambda$ is assumed to be
  sufficiently large.  Moreover, in \cite{Besch2000} stronger
  regularity assumptions on $b$ are made. In addition, the magnetic
  vector potential ${\bf a}$ associated to $b$ is assumed to decay at
  infinity. However, the results of \cite{Besch2000} hold for more
  general background magnetic fields than $B_0$. We note also that our
  proof differs from the one in \cite{Besch2000}.
\end{remark} 
\begin{remark}
Assume that (A1) and (A2) are fulfilled. As a consequence of lemmas \ref{essentialda} and \ref{Cor} below,
\begin{align*}
  \sigma_{\rm ess}(H)= \sigma_{\rm ess}(\DA)=\sigma_{\rm ess}(D_{{\bf A}_0})\,,
\end{align*}
for any ${\bf A}\in L^p_{\rm loc}(\R^2;\R^2),\,p\in(2,\infty],$ with ${\rm curl}\, {\bf A} = B$.
\end{remark}
Our next result state that eigenfunctions corresponding to the discrete spectrum of $H$ are super-exponentially localized.
\begin{theorem}\label{Maintheorem}
Assume that $B$ and $V$ satisfy (A1) and (A2) respectively and let ${\bf A}\in L^p_{\rm loc}(\R^2;\R^2),\,p\in(2,\infty],$ with ${\rm curl}\, {\bf A} = B=B_0+b$. Then, for any eigenfunction $\Psi$ of $H=\DA+V$ with $H\Psi=E\Psi$ and 
$E\in \R\setminus\sigma(D_{{\bf A}_0})$ the following holds:

For every $r\in[2,\infty]$ and $\gamma>0$ there exists an $R>0$ such that
\begin{equation}\label{sexpon}
\|\mathbbm{1}_{\{|\bx|\ge R\}}e^{\gamma |\bx|}\Psi\|_r<\infty\,.
\end{equation}
\end{theorem}
This theorem is proven in Section \ref{superex}.
\begin{remark}
This type of results are known to hold for magnetic Schr\"odinger operators
$(\bp-{\bf A})^2+B$. Our proof   follows the ideas presented in \cite{CorneanNenciu1998}.
In fact, since our operator is linear in ${\bf A}$, some parts of the argument are more straightforward. For instance, we do not require that $b\in C^1(\R^2;\R)$  decays in the $C^1$-norm as  done in \cite{CorneanNenciu1998}.
\end{remark}
\begin{remark}
One  essential ingredient in the proof of Theorem \ref{Maintheorem} is 
the explicit knowledge of the  Green function $G_0$   of $D_{{\bf A}_0}$. This is calculated in  Appendix~\ref{greenf}.
\end{remark}
In order to obtain  Gaussian decay we make further assumptions on $(B,V)$. Let $T=\R/(2\pi \Z)$ and let $v=v(r,\theta),(r,\theta)\in\R^+\times T$ be the potential $V$ written in
polar coordinates. We assume:
%
%
\begin{enumerate}
\item[(A3)] $B$  is radially symmetric, i.e., $b(\bx)=b(r),\ r=|\bx|$. 
\item[(A4)] For any $(r,\theta)\in
  \R^+\times T$ the mapping $\R\ni a \mapsto v(r,\theta+a)=:{v}_a(r,\theta)$ has an analytic continuation $\widetilde{v}_z(r,\theta)$ to $\C$. Moreover, for any $\tau>0$ there exist a $p\in(2,\infty]$ and a real-valued function $u_\tau\in L_{\rm loc}^p(\R^+\times T,r\rd r\,\rd \theta)$ such that $\|\mathbbm{1}_{\{r>n\}}u_\tau\|_\infty\to 0$  as $n\to\infty$ and 
$$|\widetilde{v}_z(r,\theta)|\le u_\tau (r,\theta)\,,$$
for any  $(r,\theta)\in
  \R^+\times T$ and $z\in S_\tau:= \{z \in \C\,:\, |\operatorname{Im}
  z|\le \tau \}$.
\item[(A5)] $v$ is differentiable with respect to $r$ and  $\R\ni a \mapsto \partial_rv(r,\theta+a)$ can be  analytically  continued to $\partial_r\widetilde{v}_z(r,\theta)$ on $\C$. Moreover, there exist a $\rho>0$ such that for any $\tau>0$ there is  $\kappa_\tau>0$ such that $|\mathbbm{1}_{\{r>\rho\}} \partial_rv_z(r,\theta)|\le \kappa_\tau $ for   any  $(r,\theta)\in
  \R^+\times T$ and $z\in S_\tau$.
\end{enumerate}
\begin{theorem}\label{Maintheorem3}
Assume that $B$ satisfies (A1) and (A3) and $V$  satisfies (A2),(A4) and (A5). Let  ${\bf A}\in L^p_{\rm loc}(\R^2;\R^2),\,p\in(2,\infty],$ with ${\rm curl}\, {\bf A} = B$.  Then, for any eigenfunction $\Psi$ of $H=\DA+V$ with $H\Psi=E\Psi$ and 
$E\in \R\setminus\sigma(D_{{\bf A}_0})$ the following holds:
 For every $0<\alpha<1$,  we have 
\begin{equation}
\|e^{\alpha B_0/4|\bx|^2}\Psi\|_2<\infty\,.
\end{equation}
\end{theorem}
This theorem is proven in Section \ref{Gauss}.
\begin{remark}
The analyticity assumption (A4)  on the angular variable of $V$  implies, by a Paley-Wiener argument, exponential decay of the
Fourier modes of the potential in its angular momentum decomposition
 (see equations \eqref{defW}, \eqref{FourierKoeff} and \eqref{decV} below). The assumption (A5) is similar to (A4) but for the radial derivative of the potential. 
\end{remark}
\begin{remark}
  The first proof of Gaussian localization for magnetic Schr\"odinger
  operators using assumptions like (A4) (but not (A5)) was given in
  \cite{Erdos}. In addition, an example of a potential decaying at
  infinity for which the corresponding ground state decays slower than
  a Gaussian is also given in \cite{Erdos}.  The proof in \cite{Erdos}
  is based on a generalized Feynman-Kac formula.  An alternative proof
  using Agmon-type estimates with localizations in space and angular
  momentum was given in \cite{Nakamura1996}. A variation of the method
  in \cite{Nakamura1996} was used in \cite{Sordoni1998} to treat the
  general $n$-dimensional case, again for magnetic Schr\"odinger
  operators. Our proof follows the ideas developed in
  \cite{Nakamura1996}. However, it turns out to be more involved since
  our operator is not bounded from below.  To overcome this difficulty
  we square the Dirac operator (or parts of it). This is the reason
  why  (A5) is used in our setting.
\end{remark}
\noindent
{\it The article is organized as follows}: In Section \ref{preliminars}
we review some essentially well known facts about 
magnetic Dirac operators. Sections \ref{spectrum}, \ref{superex}, and
\ref{Gauss} are devoted to the proofs of theorems \ref{Teo1},
\ref{Maintheorem}, and \ref{Maintheorem3} respectively. The article 
ends with an appendix containing some useful technical results.


\bigskip

\noindent
{\bf Acknowledgements.}
E.S. thanks Horia Cornean for stimulating discussions in the conference `Spectral Days' in Santiago.
Both authors have been partially supported by the DFG (SFB/TR12).

\section{Preliminaries}\label{preliminars}
\subsection{Essential self-adjointness}\label{def-of-H}
Throughout this article we consider magnetic
potentials ${\bf A}\in L^p_{{\rm
   loc}}(\R^2;\R^2)$ and electric potentials $V\in L^q_{{\rm
   loc}}(\R^2),\ p,q\in(2,\infty]$.
In order to show essential self-adjointness of the operators $H$ and $D_{\bf A}$ defined in \eqref{eq:1} and \eqref{leila} it suffices
to prove that 
$$ H_R\phi:= {\boldsymbol \sigma}\cdot(\bp-\mathbbm{1}_{\{|\bx|\le R\}}{\bf A})\phi+ \mathbbm{1}_{\{|\bx|\le R\}}V\phi\,,\qquad \phi\in C_0^\infty(\R^2;\C^2)\,,$$
is essentially self-adjoint  for every $R>0$ (see \cite{Chernoff77}). 
Using that for $f\in L^p(\R^2;\C)$ and $2<p<\infty$
\begin{equation*}\label{Intop}
f(\bx)(\bp^2+1)^{-1/2}
\end{equation*}
is a compact operator  (see \cite[Theorem 4.1]{Simon2005}) we get
that $\mathbbm{1}_{\{|\bx|\le R\}}(V-{\boldsymbol \sigma}\cdot \bf A)$ is a relative compact perturbation of $D_{\bf 0}$.
This shows essential self-adjointness of $H_R$, since $D_{\bf 0}$ is essentially self-adjoint on $C_0^\infty(\R^2;\C^2)$.
\subsection{Gauge invariance}\label{SecGauge}
Let ${\bf A}, \hat{\bf A}\in L^p_{{\rm
   loc}}(\R^2;\R^2),\ 2< p<\infty,$ be two vector potentials  with
$$
{\rm curl}\, {\bf A}={\rm curl}\, \hat{\bf A}
$$
in the sense of distributions. According  to \cite{Leinfelder1983}  there is a gauge function $\hat{\Phi}\in W^{1,p}_{loc}(\R^2;\R)$ such that
$$
{\bf A}=\hat{\bf A}+\nabla \hat{\Phi}\,.
$$
It follows, for any electric potential $V\in L^q_{{\rm
   loc}}(\R^2;\R),\, q\in(2,\infty]$, that
$$ (D_{\bf A}+V) =e^{\ri \hat{\Phi}} (D_{\hat{\bf A}}+V)e^{-\ri \hat{\Phi}}\,.$$
In particular, $e^{\ri \hat{\Phi}} (D_{\hat{\bf A}}+V)e^{-\ri \hat{\Phi}}$ is essentially self-adjoint on $C_0^\infty(\R^2;\C^2)$. 

\noindent
This can be seen as follows: 
Note that $\mathcal{D}(e^{\ri \hat{\Phi}} (D_{\hat{\bf A}}+V)e^{-\ri \hat{\Phi}})=\{ 
f\in L^2(\R^2;\C^2)\,:\, e^{-\ri \hat{\Phi}}f\in \mathcal{D}(D_{\hat{\bf A}}+V)\}$.
Pick functions $\eta,\eta'\in {C}^\infty_0(\R^2;\C^2)$ and a sequence $(\hat{\Phi}_{m})_{m\in\N}$
 in ${C}^\infty(\R^2;\R)$ with $\hat{\Phi}_{m}\to \hat{\Phi}$ in $W^{1,p}_{loc}(\R^2)$ (and hence in $W^{1,2}_{loc}(\R^2)$) as $m\to\infty$. Then,
\begin{align*}
\sps{(D_{\hat{\mathbf{A}}}+V)\eta'}{e^{-\ri \hat{\Phi}}\eta}
&=\lim_{m\to\infty}\sps{(D_{\hat{\mathbf{A}}}+V)\eta'}{e^{-\ri \hat{\Phi}_{m}}\eta}\\
&=\lim_{m\to\infty}\sps{e^{\ri \hat{\Phi}_{m}}\eta'}{(D_{\hat{\mathbf{A}}}+V)\eta}-\lim_{m\to\infty}\sps{e^{\ri \hat{\Phi}_{m}}\eta'}{{\boldsymbol \sigma}\cdot\nabla\hat{\Phi}_{m}\eta}\\
&=\sps{e^{\ri \hat{\Phi}}\eta'}{(D_{\hat{\mathbf{A}}}+V)\eta}-\sps{e^{\ri \hat{\Phi}}\eta'}{{\boldsymbol \sigma}\cdot\nabla\hat{\Phi}\eta}\,.
\end{align*}
Since $\eta'$ is an arbitrary element of a core of $D_{\hat{\mathbf{A}}}+V$, it follows that $e^{-\ri\hat{\Phi}}\eta\in \mathcal{D}(D_{\hat{\mathbf{A}}}+V)$
and $(D_{\hat{\mathbf{A}}}+V)e^{-\ri\hat{\Phi}}\eta=e^{-\ri \hat{\Phi}}(D_{\hat{\mathbf{A}}}+V-{\boldsymbol \sigma}\cdot\nabla\hat{\Phi})\eta$
which implies that
\begin{align*}
e^{\ri \hat{\Phi}}(D_{\hat{\mathbf{A}}}+V)e^{-\ri\hat{\Phi}}\eta=(D_{{\mathbf{A}}}+V)\eta\,,\qquad \eta\in {C}^\infty_0(\R^2;\C^2)\,.
\end{align*}
Due to the essential self-adjointness of  $D_{\mathbf{A}}+V$  we deduce that
$e^{\ri\hat{\Phi}}(D_{\hat{\mathbf{A}}}+V)e^{-\ri\hat{\Phi}}$ is also essentially self-adjoint on ${C}^\infty_0(\R^2;\C^2)$ and that the two operators coincide.
\subsection{Supersymmetry}\label{supers}
For ${\bf A}=(A_1,A_2)$ with $A_j\in L^p_{{\rm loc}}(\R^2),
\,p\in(2,\infty]$, $j=1,2$, we define the following two operators
\begin{align*}
 d_1\phi= [(p_1-A_1)+\ri (p_2-A_2)]\phi\,,\qquad \phi \in C_0^\infty(\R^2;\C)\,,\\
d_2\phi=[(p_1-A_1)-\ri (p_2-A_2)]\phi\,,\qquad \phi \in C_0^\infty(\R^2;\C)\,.
\end{align*}
Clearly, we have that
\begin{align*}
 D_{\bf A}\upharpoonright_{C_0^\infty(\R^2;\C^2)}=
\left(
\begin{array}{cc}
0 & d_2\\
d_1 & 0
\end{array}\right)\,.
\end{align*}
Since $ D_{\bf A}\upharpoonright_{C_0^\infty(\R^2;\C^2)}$ is essentially self-adjoint it follows that $d_1$ and $d_2$ are closable \cite[Section 5.2.2]{Thaller}. In addition, 
setting $d:=\overline{d_1}$ one finds that $d^*=\overline{d_2}$ and  
\begin{equation}
  \label{eq:4}
  D_{\bf A}=
\left(
\begin{array}{cc}
0 & d^*\\
d & 0
\end{array}\right)\,\quad\mbox{on}\quad \mathcal{D}(D_{\bf A})=
\mathcal{D}(d)\oplus\mathcal{D}(d^*)\,.
\end{equation}
It is known that $dd^*$ and $d^*d$ are self-adjoint with domains $\mathcal{D}(dd^*)=\{ \phi\in \mathcal{D}(d^*)\,:\,d^*\phi\in \mathcal{D}(d)\}$ and $\mathcal{D}(d^*d)=\{ \phi\in \mathcal{D}(d)\,:\,d\phi\in \mathcal{D}(d^*)\}$. Moreover, there is a unitary map $S$ from $\rm Ker(dd^*)^\bot$ to $\rm Ker(d^*d)^\bot$, such that
\begin{equation}\label{Unitary}
 dd^*\upharpoonright_{\rm Ker(dd^*)^\bot} = S^* d^*d\upharpoonright_{\rm Ker(d^*d)^\bot}  S\,.
\end{equation}
Let us note that we can block-diagonalize  $\DA$ using the
Foldy-Wouthuysen transformation. Setting 
\begin{equation*}
  a_+=\left\{
\begin{array}{ll}
1/\sqrt{2} &\,\mbox{on}\,\,\,\mathrm{Ker}(D_{\bf A})^\perp\\
1 &\,\mbox{on}\,\,\,\mathrm{Ker}(D_{\bf A})
\end{array}\right.,\,\,\,\,\,\,
a_-=\left\{
\begin{array}{ll}
1/\sqrt{2} &\,\mbox{on}\,\,\,\mathrm{Ker}(D_{\bf A})^\perp\\
0 &\,\mbox{on}\,\,\,\mathrm{Ker}(D_{\bf A})
\end{array}\right.,
\end{equation*}
we define the  Foldy-Wouthuysen transformation  as
\begin{equation*}
  U=a_++\sigma_3 \mathrm{sgn}(D_{\bf A}) a_-\,,
\end{equation*}
where $\mathrm{sgn}(D_{\bf A})=\DA/|\DA|$ on $\mathrm{Ker}(D_{\bf A})^\perp$ and equals zero on $\mathrm{Ker}(D_{\bf A})$ and 
\begin{equation*}
 \sigma_3=
\left(
\begin{array}{cc}
1&0\\
0&-1
\end{array}\right)\,.
\end{equation*}
The unitarity of the above transformation can be easily verified observing
that $\varphi\in \mathrm{Ker}(D_{\bf A}) \Leftrightarrow
\sigma_3\varphi\in \mathrm{Ker}(D_{\bf A})$ and that $\sigma_3\mathrm{sgn}(D_{\bf A})=-\mathrm{sgn}(D_{\bf A})\sigma_3$. The latter relation holds since $\sigma_3D_{\bf A}=-D_{\bf A}\sigma_3$ and $\sigma_3|D_{\bf A}|=|D_{\bf A}|\sigma_3$. A direct computation yields
\begin{equation}\label{Isom}
  U D_{\bf A} U^{*}=\left(
\begin{array}{cc}
\sqrt{d^* d} & 0\\
0 & -\sqrt{d d^*}
\end{array}\right)\,.
\end{equation}
Equation \eqref{Unitary} and \eqref{Isom} imply the following statement.
\begin{proposition}\label{supersymmetry} 
Let ${\bf A}\in L^p_{{\rm
   loc}}(\R^2;\R^2)$ for some $p\in(2,\infty]$. Then, the spectrum of $\DA$  is symmetric with respect to zero and 
\begin{equation*}\label{www}
\sigma_{\sharp}(\DA)\cap (0,\infty)= \sigma_{\sharp}(\sqrt{d^*d})\setminus \{0\}\,,
\qquad \sharp\in\{\rm pp, d, ess\}\,.
\end{equation*} 
\end{proposition}
\section{The spectrum of $\DA$}\label{spectrum}
The aim of this section is to show Theorem \ref{Teo1}.  An important
ingredient is the study of the essential spectrum of $\DA$.  In order
to do that we  modify an argument from \cite{Iwatsuka1983}
obtaining Lemma \ref{essentialda} below. We combine this  with a result
from \cite{Rozenblum2006} on the infiniteness of
zero modes for Pauli-operators (see Lemma \ref{Cor} below). The proof of the theorem is then a consequence of lemmas \ref{Cor} and \ref{lemma}.

In the following discussion we assume that $B=B_0+b$ with $B_0>0$ and 
 $b\in L^1_{\rm loc}(\R^2;\R)$ such that $|b|^{1/2}$ is relative $\sqrt{\bp^2+1}$- compact. Let ${\bf A}\in L^p_{\rm loc}(\R^2;\R^2),\,p\in (2,\infty],$ with ${\rm curl}\, {\bf A} = B$.
We start by  observing that, for $\phi\in C_0^\infty(\R^2;\C)$, 
\begin{equation}\label{w1}
\begin{split}
\sps{d^*\phi}{d^*\phi}&= \sum_{j=1}^2 \| (p_j-A_j)\phi\|^2+ \sps{\phi}{B \phi}\\
\sps{d\phi}{d\phi}&= \sum_{j=1}^2 \| (p_j-A_j)\phi\|^2- \sps{\phi}{B \phi}
\end{split}
\end{equation}
holds. This implies the commutator relation
\begin{align}
  \label{CommutatorB}
\sps{\phi}{[d,d^*]\phi}:=
\sps{d^*\phi}{d^*\phi}-\sps{d\phi}{d\phi}=2\sps{\phi}{ B\phi}\,,\quad\phi\in C_0^\infty(\R^2;\C)\,.
\end{align} 
The idea in \cite{Iwatsuka1983} is to use this commutator  to study
the essential spectrum of  $dd^*$ and $d^*d$.
In order to extend this identity we define these operators 
as quadratic forms and show that
$\mathcal{Q}(b)\supset \mathcal{Q}(d^*d)=\mathcal{Q}(dd^*)$
and $|b|^{1/2}( d^*d +1)^{-1/2}$ is a compact operator. Here $\mathcal{Q}(\cdot)$ is used to denote the form domain.

Let us  define 
\begin{align*}
q_1(\phi,\phi)={q}_1[\phi]:= \|d\phi\|^2,
\qquad q_2(\phi,\phi)={q}_2[\phi]:=\|d^*\phi\|^2,
\end{align*}
with form domains $\mathcal{Q}(q_1)=\mathcal{D}(d)$ and $\mathcal{Q}(q_2)=\mathcal{D}(d^*)$. Since $d$ and $d^*$ are closed (see Subsection \ref{supers}) we have that $q_1$ and $q_2$ are closed and positive. Thus, associated to $q_j$, $j=1,2$, there is a unique self-adjoint operator $T_j$ characterized as follows:
\begin{equation}
  \label{eq:3}  
\begin{split}
\sps{\psi}{T_j \varphi}=q_j(\psi,\varphi),\quad \psi\in \mathcal{Q}(q_j), \varphi\in\mathcal{D}(T_j)\,,\\
\mathcal{D}(T_j)=\{\varphi\in Q(q_j)\,|\, \exists \eta\in L^2(\R^2;\C), \forall \psi\in \mathcal{C}, q_j(\psi,\varphi)=\sps{\psi}{\eta}\}\,,
\end{split}
\end{equation}
where $\mathcal{C}$ is any form core of $q_j$. It is easy to check using \eqref{eq:3} that in fact $T_1=d^* d$ and $T_2=d d^*$.
Note that since the restrictions of $d$ and $d^*$ to $C_0^\infty(\R^2,\C)$ are closable $C_0^\infty(\R^2,\C)$ is a form core for $q_1$ and $q_2$.
We define yet another quadratic form. For  $\phi\in C_0^\infty(\R^2,\C)$ we set
\begin{align*}
\tilde{q}_3[\phi]:=\sum_{j=1}^2\|(p_j-A_j)\phi\|^2\,.  
\end{align*}
It is known \cite{Simon1979} that $\tilde{q}_3$ is closable 
and we denote its closure by $q_3$. Its associated self-adjoint operator $H_S=:(\bp-{\bf A})^2$ is the usual magnetic Schr\"odinger operator. Recall that $|b|^{1/2}$ is relative $\sqrt{\bp^2+1}$-compact. Using the 
diamagnetic inequality for $|\bp -{\bf A}|$ (see e.g. \cite{Franketal2007}) and arguing as in \cite[Theorem 2.6]{AHS-I} we conclude that $\mathcal{Q}(q_3)=\mathcal{D}(H_S^{1/2})\subset\mathcal{D}(|b|^{1/2})$ and that
$|b|^{1/2}$ is relative $H_S^{1/2}$-compact. Thus, the quadratic form 
\begin{align*}
  \beta[\varphi]:=B_0\|\varphi\|^2+
\sps{{\rm sgn}(b)|b|^{1/2}\varphi}{|b|^{1/2}\varphi}
\end{align*}
is in absolute value bounded with respect to $q_3$ with bound $0$. In particular,
\begin{align*}
  q_3^{\pm}[\varphi]:=q_3[\varphi]\pm\beta[\varphi]\,,\quad \varphi\in\mathcal{Q}(q_3)\,,
\end{align*}
 is closed. Observing that by \eqref{w1} we have that $q_2\upharpoonright_{C_0^\infty}=q_3^+\upharpoonright_{C_0^\infty}$
and $q_1\upharpoonright_{C_0^\infty}=q_3^-\upharpoonright_{C_0^\infty}$ and
using that $C_0^\infty(\R^2;\C)$ is a form core for $q_1, q_2, q_3 $ and $ q_3^{\pm}$ we conclude that $\mathcal{Q}(q_1)=\mathcal{Q}(q_2)=\mathcal{Q}(q_3)\equiv \mathcal{Q}$ 
and $q_1=q_3^-$ and $q_2=q_3^+$. Moreover,
\begin{equation}\label{dddd}
\begin{split}
  d d^*&= (\bp-{\bf A})^2+B\,\\
  d^* d&=(\bp-{\bf A})^2-B
\end{split}
\end{equation}
in the sense of quadratic forms on $\mathcal{Q}$ and hence the commutator formula \eqref{CommutatorB} extends to $\mathcal{Q}$.

\begin{lemma}\label{essentialda}
Let $B=B_0+b$ with $B_0>0$ and $|b|^{1/2}\in L^2_{\rm loc}(\R^2;\R)$ be relative $\sqrt{\bp^2+1}$-compact.
Let ${\bf A}\in L^p_{\rm loc}(\R^2;\R^2), p\in (2,\infty]$, with
${\rm curl} \,{\bf A} =B.$  Then,  either one of the following statements holds
\begin{itemize}
\item [i)] $\sigma_{ess}(d^*d)=\emptyset\qquad$ 
\item [ii)] $\sigma_{ess}(d^*d)=\{ 2B_0 n\,:\, n\in \N_0\}$ and $\sigma_{ess}(dd^*)=\{ 2B_0 n\,:\, n\in \N\}$\,.
\end{itemize}
In addition, if $V$ satisfies (A2) then $V$ is relative $\DA$- compact
and in particular $\sigma_{\rm ess}(\DA)=\sigma_{\rm ess}(\DA+V)$.
\end{lemma}
\begin{remark}\label{poisson}
Note that our assumption on $B$ are satisfied if $B$  fulfills (A1). Indeed, in this case  $|b|^{1/2}(\bp^2+1)^{-1/2}$ is compact by Lemma \ref{compsi} in Appendix \ref{family}. 

Moreover, note that if $B\in L_{\rm loc}^{q}(\R^2;\R)$ for some $q>1$ we can always find  ${\bf A}\in L^p_{\rm loc}(\R^2;\R^2)$ for some  $p\in (2,\infty]$. In order to see this define
$h$ to be a solution of 
\begin{equation}\label{newton}
\Delta h=B\,.
\end{equation}
A local solution to this equation is given by the Newton potential $h_N$ of $B$. We know that $h_N\in W^{2,q}(\Omega)$ by the Calderon-Zygmund inequality, where $\Omega\subset \R^2$ is a bounded domain (see e.g. \cite[Sec.9.4]{GT1983}). This property extends to any solution $h$ of \eqref{newton} since  $h-h_N$ is harmonic on $\Omega$. Therefore, $h\in W_{\rm loc}^{2,q}(\R^2)$. Now one can define ${\bf A}:=(-\partial_2 h,\partial_1 h)$. Clearly, $A_j\in  W_{\rm loc}^{1,q}(\R^2)$. By standard Sobolev inequalities one obtains that $A_j\in L_{\rm loc}^t(\R^2;\R)$ for some $2<t<\infty$
if $q\in (1,2]$   and $A_j\in L_{\rm loc}^\infty(\R^2;\R)$ if $q> 2$.
\end{remark}
\begin{proof}
First note that for any $\lambda\ge 0$ the 
operator $(d^*d+2B_0+\lambda)^{-1/2}$ maps $L^2(\R^2;\C)$ onto $\mathcal{D}(\sqrt{d^*d})$ which equals $\mathcal{Q}$ and $\mathcal{D}(H_S^{1/2})$. Thus, by the closed graph theorem, the operator 
$(H_S+1)^{1/2}(d^*d+2B_0+\lambda)^{-1/2}$ is bounded. In particular,
\begin{align*}
|b|^{1/2}(d^*d+2B_0+\lambda)^{-1/2}=|b|^{1/2}(H_S+1)^{-1/2}(H_S+1)^{1/2}(d^*d+2B_0+\lambda)^{-1/2}
\end{align*}
is compact. Hence, the operator
\begin{align*}
  T(\lambda):=\overline{(d^*d+2B_0+\lambda)^{-1/2}{\rm sgn}(b)|b|^{1/2}} |b|^{1/2}(d^*d+2B_0+\lambda)^{-1/2}\,
\end{align*}
is also compact. It is easy to see that $\lambda>0$ can be chosen so large that $\|T(\lambda)\|<1$. For such $\lambda$'s we have, according to the resolvent formula for operators defined as quadratic forms (see \cite{Simon1971}),  that
\begin{align*}
  (d^*d+2B_0+2b+\lambda)^{-1}=
(d^*d+2B_0+\lambda)^{-1/2}(1+T(\lambda))^{-1}(d^*d+2B_0+\lambda)^{-1/2}\,.
\end{align*}
Note that the inverse of $1+T(\lambda)$ is well defined as a geometric expansion. Since $(1+T(\lambda))^{-1}-1$ is compact, we conclude that the resolvent difference between $d^*d+2B_0+2b+\lambda$ and $d^*d+2B_0+\lambda$
is also compact. Therefore, by Weyl's theorem, the two operators have the same essential spectrum. Using this and \eqref{dddd} we deduce that
\begin{align}\label{eess}
  \sigma_{\rm ess}(dd^*)=\sigma_{\rm ess}(d^*d+2B_0+2b)=\sigma_{\rm ess}(d^*d+2B_0)\,.
\end{align}
The latter equality and Equation \eqref{Unitary} imply (here we follow \cite{Iwatsuka1983})
\begin{equation}
  \label{iwa1}
\begin{split}
  &S:=\sigma_{\rm ess}(d^*d)\,,\qquad S\subset[0,\infty)\,,\\
&S\setminus\{0\}=S+2B_0\,.
\end{split}
\end{equation}
Assume now that $S\not=\emptyset$, then it is easy  to see from \eqref{iwa1}  that $0\in S$ and hence $2B_0n\in S, \,n\in\N_0$. Note also that no other points can belong to $S$. Hence, 
using \eqref{eess} we get that  $ \sigma_{\rm ess}(dd^*)=2B_0n,\, n\in\N$.

Now, assume that $V$ fulfills (A2). Then, $V$ is relative
$\sqrt{\bp^2+1}$- compact (see Lemma \ref{compsi} in Appendix
\ref{family}). It follows by the diamagnetic inequality that $V$ is
relative $H_S^{1/2}$- compact and consequently (arguing as before for
$b$) $\mathcal{D}(V)\supset \mathcal{Q}$ and the operators
$V(dd^*+\lambda^2)^{-1/2}$ and $V(d^*d+\lambda^2)^{-1/2}$ are compact
for any $\lambda\not=0$. From these considerations follow that
$V(D_{\bf A}-\ri \lambda)^{-1}$ is compact, since the identity
  \begin{align*}
    (D_{\bf A}-\ri \lambda)^{-1}&=(\DA^2+\lambda^2)^{-1/2}[(\DA^2+\lambda^2)^{-1/2}(\DA+\ri\lambda)]\\
&=\left(\begin{array}{cc}
(d^*d+\lambda^2)^{-1/2} & 0\\
0 &(dd^*+\lambda^2)^{-1/2}
\end{array}\right)\cdot[(\DA^2+\lambda^2)^{-1/2}(\DA+\ri\lambda)]
  \end{align*}
holds and the operator in $[\dots]$ is bounded. Therefore,
$$ \sigma_{ess}(D_{{\bf A}}+V)=\sigma_{ess}(D_{{\bf A}}).$$
\end{proof}
We note that if $b$ satisfies (A1) then ${\rm Ker }(d^*d)$ is infinitely degenerated. Indeed, this follows from the fact that
\begin{equation}
\int_{\R^2}[B]_+d^2x =\infty\,,\qquad  \int_{\R^2}[B]_-d^2x <\infty\,,
\end{equation}
(where $[f]_+$  and $[f]_-$ are the positive and negative parts of $f$) which
shows that $B=B_0+b$ fulfills the conditions of \cite[Corollary 3.4]{Rozenblum2006}. In particular, we know that
\begin{align*}
  {\rm Ker }(d^*d)= \{ \omega e^{-h}\,|\, \omega e^{-h}\in L^{2}(\R^2;\C)\,,\, \omega \,\,\mbox{is analytic in}\,\,x_1+\ri x_2\}\,,
\end{align*}
where $h$ is a solution of the equation $\Delta h=B$ \cite{Rozenblum2006}.
Therefore, we get:
\begin{lemma}\label{Cor}
Assume that $B$ satisfies (A1) and let ${\bf A}\in L^p_{\rm loc}(\R^2;\R^2),\,p\in(2,\infty],$  with ${\rm curl}\, {\bf A} = B$. Then,
\begin{gather}\label{finaless}
\sigma_{\rm ess}(d^*d)=\{ 2B_0 n\,|\, n\in \N_0\}\,,
\qquad \sigma_{\rm ess}(dd^*)=\{ 2B_0 n\,|\, n\in \N\}\,.
\end{gather}
In particular,
\begin{align*}
 \sigma_{\rm ess}(\DA)= \sigma_{\rm ess}(D_{{\bf A}_0})=\{ l_n\,|\,n\in \Z\}\,.
\end{align*}
Moreover, $0$ is an isolated point of $\sigma(\DA)$ and $\sigma(d^*d)$.
\end{lemma}
\begin{proof}
Due to our previous discussion we see that $0\in \sigma_{\rm ess}(d^*d)$. This combined with Lemma \ref{essentialda} imply \eqref{finaless}.  That
$0$ is an isolated point of $\sigma(d^*d)$
follows  by noting that, since  $0\not \in \sigma_{ess}(dd^*)$,  $0$  is neither an accumulation point  of $\sigma(dd^*)$ nor of $\sigma(d^*d)$. The statements on $\sigma(\DA)$ are now a consequence of Proposition \ref{supersymmetry}.
\end{proof}
\begin{lemma}\label{lemma}
Assume that $B$ satisfies (A1) and let ${\bf A}\in L^p_{\rm loc}(\R^2;\R^2),\,p\in(2,\infty],$  with ${\rm curl}\, {\bf A} = B$. Then, we have:
\begin{itemize}
\item[(a)] If $b\le 0$ and strictly negative on some open set,  then 
$${\rm dim(Ran}(\mathbbm{1}_{(0,\sqrt{2B_0})}(D_{\mathbf{A}})))={\rm dim(Ran}(\mathbbm{1}_{(-\sqrt{2B_0},0)}(D_{\mathbf{A}})))=\infty\,.$$
\item[(b)] If $b\ge 0$ then 
$${\rm dim(Ran}(\mathbbm{1}_{(0,\sqrt{2B_0})}(D_{\mathbf{A}})))={\rm dim(Ran}(\mathbbm{1}_{(-\sqrt{2B_0},0)}(D_{\mathbf{A}})))=0\,.$$
\end{itemize} 
\end{lemma}
\begin{proof}
We may choose ${\bf A}:=(-\partial_2 h,\partial_1 h)$ where $h$ is a solution of 
$\Delta h=B$. Due to Remark \ref{poisson} we know that ${\bf A}\in L^p_{\rm loc}(\R^2,\R^2)$
for some $p>2$.
\noindent
Part (a): Let $\Omega$ be an open set with $b\upharpoonright \Omega<0$. 
Recall that there are infinitely many functions $\omega$, analytic in $x_1+\ri x_2$, with $\psi:=\omega e^{-h }\in {\rm Ker}(d^*d)$.
For such $\psi$ we have, using
  \eqref{dddd}, 
\begin{align}
  \label{eq:11}
 \sps{\psi}{d d^* \psi}=2\sps{\psi}{B \psi}\le 2 B_0 \|\psi\|^2+2\int_{\Omega} b(\bx)|\psi(\bx)|^2 \rd \bx <2B_0 \|\psi\|^2\,,
\end{align}
where in the last inequality we use the fact that $\psi$ can not
vanish on $\Omega$.  Let $(\psi_{n})_{n\in\N}$ be an orthonormal
system such that $\psi_n:= e^{-h}\omega_n\in {\rm Ker}\,d^*d$ with
$\omega_n$ analytic  in $x_1+\ri x_2$. For $N\in \N$ define the self-adjoint matrix $M_N :=
( \sps{ \psi_n}{dd^*\psi_m})_{1\le n,m\le N}$. 
It follows from \eqref{eq:11} that
$M_N<2B_0$. The Rayleigh-Ritz principle implies
$$0\le  \mu_n(dd^*)\le \mu_n(M_N)<2B_0,\ n=1,\ldots,N\,,$$
where we write $$\mu_n(T):= \sup_{\phi_1,\ldots,\phi_{n-1}}\inf_{\psi\in {\rm span}\{ \phi_1,\ldots,\phi_{n-1}\}^\perp \atop \|\psi\|=1,\psi\in Q(T)} \sps{\psi}{T\psi}$$
 for some self-adjoint operator T. Since $N$ is arbitrary the mini-max principle implies that
${\rm dim(Ran}(\mathbbm{1}_{[0,\sqrt{2B_0})}(dd^*)))=\infty$. It follows that
${\rm dim(Ran}(\mathbbm{1}_{(0,\sqrt{2B_0})}(dd^*)))=\infty$, for $0\notin \sigma_{\rm ess}(d d^*)$ by Lemma \ref{Cor}.
The claim is now a consequence of Proposition \ref{supersymmetry} and \eqref{Unitary}.
\\
\noindent
Part (b): In this case we have that  $dd^*\ge 2B_0$, since $dd^*-d^*d=2B\ge 2B_0$. Thus, the claim follows now from Proposition \ref{supersymmetry} and \eqref{Unitary}.
\end{proof} 
\section{Super-exponential localization}\label{superex}
The proof of Theorem \ref{Maintheorem} follows the ideas developed in \cite{CorneanNenciu1998}. An essential ingredient is that,
by means of suitable local gauge transformations on certain regions
outside a big ball of radius $n$ centered at the origin, one can
replace the operator $D_{\bf A}$ by a Dirac operator $D_{{\bf A}_n}$
with ${\bf A}_n={\bf A}_0 + {\bf a}^n$, where ${\bf a}^n$ is a
magnetic vector potential of a magnetic field $b_n$ satisfying
$\lim_{n\to\infty} \|b_n\|_{\infty}=0$. The advantage of this is that we can
obtain explicit $L^p$ estimates (see Lemma~\ref{expansion} below) for
the resolvents of $D_{{\bf A}_n}$, conjugated with exponential
weights. These estimates can be derived using a certain resolvent
expansion (see \eqref{eq:21}) in combination with an explicit
expression for the Green kernel of $D_{{\bf A}_0}$ that can be found
in Appendix \ref{greenf} below.

Before stating these $L^p$ estimates let us fix some notation. For $p,q\in[1,\infty]$ we denote by
$\mathcal{B}(p,q)$ the space of bounded operators from
$L^p(\R^2;\C^2)$ to $L^q(\R^2;\C^2)$ and write, for $T\in
\mathcal{B}(p,q) $, %
\begin{align}
  \label{eq:23}
  \|T\|_{p,q}:=\|T\|_{\mathcal{B}(p,q)}\,.
\end{align}
Let $\gamma\ge 0$ and $\bu\in\R^2$ with $|\bu|=1$. 
We define the exponential weight function as
\begin{equation*}
F(\bx):=\gamma\, \bu\cdot \bx\,,\quad \bx\in\R^2\,.
\end{equation*}
Let $b_n$ be a magnetic field with $\lim_{n\to\infty}
\|b_n\|_{\infty}=0$ and  ${\bf a}^n$ be the associated vector potential in the transversal gauge, i.e., 
\begin{equation}
  \label{eq:15}
 {\bf a}^n({\bf x}):= \int_0^1  b_n(s {\bf x})\wedge \bx \,s ds\,,
\end{equation}
where we write $a\wedge {\bf v}:= a(-v_2,v_1)$ for $a\in\R$ and ${\bf v}\in \R^2$.
The proof of the Lemma below can be found at the end of this section.
\begin{lemma}\label{expansion}
  Let $V_n \in L^\infty(\R^2;\R)$, $n\in\N$, be a family of electric
  potentials satisfying
\begin{equation*}\label{decayPotential}
\lim_{n\to\infty}\|V_n\|_\infty=0\,.
\end{equation*}
For any $n\in\N$ define the family of self-adjoint operators $D_{{\bf A}_n}+V_n$,
where ${\bf A}_n:={\bf A}_0+{\bf a}^n$ and ${\bf a}^n$ is given in \eqref{eq:15}.  Let $z\in \R\setminus\sigma(D_{{\bf
    A}_0})$ and $q,r\in[1,\infty]$ be such that
$1+\tfrac{1}{r}-\tfrac{1}{q}=\tfrac{1}{p}$ for some $p\in[1,2)$.
Then, there exists $N>0$ such that, for all $n>N$, $z\notin\sigma(D_{{\bf
    A}_n}+V_n)$ and
\begin{equation}
  \label{eq:28}
  e^{F}(D_{{\bf A}_n}+V_n-z)^{-1} e^{-F}\in \mathcal{B}(q,r)\,.
\end{equation}
\end{lemma}
In what follows we  apply the above result  to show Theorem \ref{Maintheorem}. 
\begin{proof}[Proof of Theorem \ref{Maintheorem}]
For $n\in\N$ and $\bu\in\R^2$ with $|\bu|=1$ set
\begin{equation}
\Omega_{n}=\{ \bx\in \R^2\,:\, \bu\cdot \bx > n\}\,.
\end{equation}
For $j\in\{1,2,3\}$ define $\chi_{j}\in C^\infty(\R^2;[0,1])$ with $\chi_{j}=0$ on $\R^2\setminus\Omega_{jn}$ and $\chi_{j}=1$ on $\Omega_{(j+1)n}$. We choose 
$n$ so large that
$$\|b\|_{L^\infty(\Omega_n)}<\infty\,.$$
Since $b\in L_{\rm loc}^q(\R^2),\,q>1,$ we find a vector potential ${\bf a}\in L^p_{\rm loc}(\R^2;\R^2),\, p>2,$ with ${\rm curl}\,{\bf a}=b$ (see Remark \ref{poisson}).
Define, for $\bx \in \R^2$, 
\begin{align*}
   {\bf a}^n({\bf x})= \int_0^1  b_n(s {\bf x})\wedge \bx \,s ds\,,
\end{align*}
where $b_n:=\mathbbm{1}_{\Omega_{n}}b\in L^\infty(\R^2)$. Observe that
\begin{align*}
  {\rm curl}\,{\bf a}={\rm curl}\,{\bf a}^n\quad\mbox{on}\quad\Omega_{n}\,,
\end{align*}
that $\Omega_n$ is simply connected, and that ${\bf a}^n, {\bf a}\in L^p_{\rm loc}(\R^2;\R^2)$ for some $p>2$. Therefore, there exists a gauge function
$\tilde{\Phi}_n\in W^{1,p}_{\rm loc}(\Omega_n)$ such that
 (see \cite[Lemma 1.1]{Leinfelder1983})
\begin{align}\label{total}
  \nabla \tilde{\Phi}_n= {\bf a}-{\bf a}^n\quad\mbox{on}\quad\Omega_{n}\,.
\end{align}
By multiplying $\tilde{\Phi}_n$ with a $C^\infty$- cutoff function we may define a 
$\Phi_n\in W^{1,p}_{loc}(\R^2)$ that coincides with $\tilde{\Phi}_n$ on $\Omega_{2n}$. In particular, we find that
\begin{align}
  \label{total1}
   \nabla {\Phi}_n= {\bf a}-{\bf a}^n\quad\mbox{on}\quad\Omega_{2n}\,.
\end{align}
Define now $V_n:=\chi_1V$ and observe that $\|V_n\|_\infty\to 0$ as $n\to\infty$. 
Then we get, for any $\eta\in \mathcal{C}^\infty_0(\R^2;\C^2)$ and  $j\in\{2,3\}$, using \eqref{total1} and the identity $\chi_j=\chi_1\chi_j$, 
\begin{align*}
\chi_j(D_{\mathbf{A}}+V)\eta&=\chi_j(D_{\mathbf{A}_0}-{\boldsymbol \sigma}\cdot{\bf a}+V_n)\eta\\
&=\chi_j (D_{\mathbf{A}_0}- {\boldsymbol \sigma}\cdot\nabla\Phi_n-{\boldsymbol \sigma}\cdot{\bf a}^{n}+V_n)\eta\\
&=\chi_j e^{\ri \Phi_n}(D_{\mathbf{A}_0}-{\boldsymbol \sigma}\cdot {\bf a}^{n}+V_n)e^{-\ri \Phi_n}\eta\,.
\end{align*}

Set ${\bf A}_n:={\bf A}_0+{\bf a}^n$ and let $\Psi$ be an eigenfunction of $D_{\bf A}$ with eigenvalue $E\notin \sigma(D_{\mathbf{A}_0})$. By the previous computation we obtain, for any $\eta\in \mathcal{C}^\infty_0(\R^2;\C^2)$ and  $j\in\{2,3\}$,
\begin{align*}
\langle\, e^{\ri \Phi_n}(D_{\mathbf{A}_n}+V_n-E)e^{-\ri \Phi_n}\eta\,|\,\chi_j\,\Psi\,\rangle
&=\langle \,(D_{\mathbf{A}}+V-E)\eta\,|\,\chi_j\,\Psi\,\rangle\\
&=\langle\, \ri {\boldsymbol \sigma}\cdot \nabla \chi_j \eta\,|\,\Psi\,\rangle.
\end{align*}
This equality extends to any $\eta$ in the domain of $e^{\ri \Phi_n}(D_{\mathbf{A}_n}+V_n-E)e^{-\ri \Phi_n}$,
since $ \mathcal{C}^\infty_0(\R^2;\C^2)$ is a core for $e^{\ri \Phi_n}(D_{\mathbf{A}_n}+V_n-E)e^{-\ri \Phi_n}$  (see Subsection \ref{SecGauge}). Clearly, ${\bf A}_n$ and $V_n$ satisfy the assumptions of Lemma~\ref{expansion}.  
 Thus, for $n$ sufficiently large, $E\notin \sigma(D_{\mathbf{A}_n})$ and
we may replace $\eta$ by $e^{\ri \Phi_n}(D_{\mathbf{A}_n}+V_n-E)^{-1}e^{-\ri \Phi_n}\eta$ obtaining that
\begin{equation}\label{eq:23a}
\chi_j \Psi
=-\ri e^{\ri \Phi_n}(D_{\mathbf{A}_n}+V_n-E)^{-1}e^{-\ri \Phi_n}
 ({\boldsymbol \sigma}\cdot\nabla \chi_j)\Psi,\quad j\in \{2,3\}\,.
\end{equation}
Observe that using \eqref{eq:23a} for $j=2$
in combination with Lemma~\ref{expansion} (with $q=2, r=3$, and $F=0$) we obtain  that
\begin{equation}\label{eq:25a}
\chi_2 \Psi \in L^3(\R^2;\C^2)\,.
\end{equation}
We use again \eqref{eq:23a}, for large $n$,  to get in addition that
\begin{equation}\label{eq:27}
e^{F}\chi_3 \Psi
=-\ri e^{\ri \Phi_n}\big(e^{F}(D_{\mathbf{A}_n}+V_n-E)^{-1}e^{-F}\Big)\big(e^{-\ri \Phi_n}
 e^{F}({\boldsymbol \sigma}\nabla \chi_3) \chi_2\Psi\big)\,.
\end{equation}
Since $\operatorname{supp}(\nabla\chi_3)\subset
\Omega_{3n}\setminus\Omega_{4n}$ we find thanks to \eqref{eq:25a} that
$e^{-\ri \Phi_n} e^{F}({\boldsymbol \sigma}\nabla \chi_3) \chi_2\Psi
\in L^2(\R^2;\C^2)\cap
L^3(\R^2;\C^2)$.  Thus, we may apply
Lemma~\ref{expansion} with $q=3, r=\infty$ and $q=2, r=2$ to obtain
the decay in the $L^\infty$ and $L^2$ norms respectively for $n\ge
n_0$ sufficiently large.  We obtain the desired bound \eqref{sexpon} from \eqref{eq:27}
by varying $F$ over finitely many vectors $\bu$.
\end{proof}
\begin{proof}[Proof of Lemma \ref{expansion}]
Recall that the magnetic vector potential is given by ${\bf A}_n={\bf A}_0+{\bf a}^{n}$ where ${\bf a}^{n}$ is defined in  \eqref{eq:15}.

A simple calculation shows that the vector potential
\begin{equation}
  \label{eq:12}
  {\bf a}^n_{\bf x'}({\bf x})= \int_0^1  b_n({\bf x}'+s ({\bf x}-{\bf x}'))\wedge (\bx-\bx') \,s ds\,,\quad \bx'\in\R^2\,,
\end{equation}
is also a vector potential of  the magnetic field $b_n$.  
A crucial property of ${\bf a}^n_{\bx'}$ is that
\begin{align}
  \label{eq:25}
  |{\bf a}^n_{\bx'}(\bx)|\le \|b_n\|_\infty\cdot |\bx-\bx'|, \quad \bx,\bx'\in\R^2\,.
\end{align}
Since 
$\rm{curl}\, {\bf a}^n_{\bf x'}= \rm{curl}\, {\bf a}^n$ there exists a function
$\varphi_n:\R^2\to\R$ with 
\begin{equation}
  \label{eq:13}
  \nabla_{\bx} \varphi_n(\bx,\bx')={\bf a}^n(\bx)-{\bf a}^n_{\bx'}(\bx)\,.
\end{equation}
We may further require that
\begin{equation}
  \label{eq:14}
  \varphi_n(\bx,\bx)=0.
\end{equation}
The proof of Lemma~\ref{expansion} is based upon $L^p$ estimates for the  resolvent expansion \eqref{eq:21} below. We start by defining the relevant objects and list their $L^p$ properties.  For $z\in\R\setminus\sigma(D_{{\bf A}_0})$ let   $G_0(\bx,\bx',z)$ be a representation of the Green kernel
  of $(D_{{\bf A}_0}-z)^{-1}$ as $2\times 2$-matrix. Due to \eqref{eq:17} from Appendix \ref{greenf} and the triangular inequality we obtain that
\begin{align*}
  \big\|e^{F(\bx)}G_0(\bx,\bx';z)e^{-F(\bx')}\big\|_{\C^2\otimes\C^2}\le e^{-\theta(\bx-\bx')+\gamma|\bx-\bx'|}\omega(\bx-\bx';z):=g_{F}(\bx-\bx')\,.
\end{align*}
We observe that thanks to \eqref{eq:18} we have that
\begin{equation}
  \label{eq:16}
g_{F} \in L^t(\R^2),\quad t\in[1,2)\,\quad\mbox{and} \quad|\bx|g_{F}\in  L^t(\R^2), \quad t\in[1,\infty]\,.
\end{equation}
We introduce, for $n\in \N$, the integral operators $ S_n(z),T_n(z) :
L^2(\R^2;\C^2)\to L^2(\R^2;\C^2)$ with
\begin{align}
  \label{eq:20}
  \big(S_n(z)f\big)(\bx)&:=
\int_{\R^2}e^{\ri\varphi_n(\bx,\bx')}G_0(\bx,\bx';z) f(\bx')  d\bx',\\
  \big(T_n(z)f\big)(\bx)&:=
\int_{\R^2}{\boldsymbol \sigma}\cdot a^n_{\bx'}(\bx)\,e^{\ri\varphi_n(\bx,\bx')}G_0(\bx,\bx';z)f(\bx')  d\bx'\,,
\end{align}
where $\varphi_n$ is determined by \eqref{eq:13} and \eqref{eq:14}.
Notice that in view of \eqref{eq:16}, \eqref{eq:25},
and Young's inequality (see \cite[Section 4.2]{LL2001})  both operators are well defined and bounded. In fact, since 
$$
\|(e^{F} T_n(z) e^{-F})(\bx,\bx')\|_{\C^2\otimes\C^2}\le \|b_n\|_\infty |\bx-\bx'|
g_{F}(\bx-\bx')\,,
$$
we find by \eqref{eq:16} and Young's inequality that, for $q\in[1,\infty]$,
\begin{equation}
    \label{eq:24}
\lim_{n\to\infty} \|e^F T_n(z) e^{-F}\|_{q,q}=0\,.
  \end{equation}
Furthermore, a similar argument implies that, for $q,r\in[1,\infty]$ and 
$t\in [1,2)$ with  $\tfrac{1}{t}=1+\tfrac{1}{r}-\tfrac{1}{q}$,
\begin{align} \label{eq:24-b}  
\sup_{n\in \N}\|e^{F}S_n(z)e^{-F}\|_{q,r}<\infty\,.
\end{align}
Our next task is to show the following resolvent
formula in $L^2(\R^2;\C^2)$, for $n\in\N$ so large that $\|T_n(z)\|_{2,2}<1$ (see \eqref{eq:24}),
\begin{equation}
  \label{eq:21}
  (D_{{\bf A}_n}-z)^{-1}= S_n(z)\sum_{k=0}^\infty T_n(z)^k\,.
\end{equation}
Pick functions $f\in L^2(\R^2;\C^2)$ and $g\in C^\infty_0(\R^2;\C^2)$ and, we find
\begin{align*}
\label{eq:22}
&\SP{(D_{{\bf A}_n}-z)g}{S_n(z) f}\\
&\quad=\int_{\R^2}\SP{\big[(D_{{\bf A}_n}-z)g](\bx)}{\int_{\R^2} e^{\ri\varphi_n(\bx,\bx')} G_0(\bx,\bx';z)f(\bx')\,d\bx'}_{\C^2}\,d\bx\\
&\quad=\int_{\R^2}\int_{\R^2}\SP{e^{-\ri \varphi_n(\bx,\bx')}
\big[(D_{{\bf A}_n}-z)g](\bx)}{  G_0(\bx,\bx';z)
f(\bx')}_{\C^2}\,d\bx\,d\bx'\,,
\end{align*}
where Young's inequality together with Lemma \ref{GreenKernel} from Appendix \ref{greenf} enabled us to 
use Fubini's theorem in the last equality.  
Observe that due to \eqref{eq:13}
\begin{align*}
  e^{-\ri \varphi_n(\bx,\bx')}
  \big[(D_{{\bf A}_n}-z)g](\bx)
  =\big[(D_{{\bf A}_0}-{\boldsymbol \sigma}\cdot {\bf a}^n_{\bx'}-z) 
  e^{-\ri \varphi_n(\cdot,\bx')}g](\bx)\,.
\end{align*}
Hence, using again Fubini's theorem,
\begin{equation}\label{ide}
\begin{split}
  &\SP{(D_{{\bf A}_n}-z)g}{S_n(z) f}=\\
&\quad\int_{\R^2}\SP{ \int_{\R^2} G_0(\bx',\bx;z)
\big[(D_{{\bf A}_0}-z)e^{-\ri \varphi_n(\cdot,\bx')}g](\bx)\, d\bx}{ f(\bx')}_{\C^2}\,d\bx'\\
&-\int_{\R^2}\int_{\R^2}\SP{ {\boldsymbol \sigma}\cdot{\bf a}^n_{\bx'}(\bx) 
  e^{-\ri \varphi_n(\bx,\bx')}g(\bx)}{  G_0(\bx,\bx';z)f(\bx')}_{\C^2}\,d\bx'\,d\bx\\
&=\SP{g}{ f}- \SP{g}{T_n(z) f}\,.
\end{split}
\end{equation}
For the first integral after the first equality 
above we used \eqref{eq:14} together with the fact 
that $G_0$ is the Green kernel of $D_{{\bf A}_0}$ 
and thus, for any $\tilde{g}\in C^\infty_0(\R^2;\C^2)$,
\begin{align*}
  \int_{\R^2} G_0(\bx',\bx;z)\big[(D_{{\bf A}_0}-z)\tilde{g}](\bx)\,d\bx
  =\tilde{g}(\bx')\quad \mbox{a.e.}\,.
\end{align*}
Now,  since  $D_{{\bf A}_n}$ is essentially self-adjoint on $C^\infty_0(\R^2;\C^2)$ 
we can extend the identity \eqref{ide} for all $g\in \mathcal{D}(D_{{\bf A}_n})$. 
From this extension  follows that $S_n(z)$ maps $L^2(\R^2;\C^2)$ in   
$\mathcal{D}(D_{{\bf A}_n})$ and 
\begin{equation*}
  \label{eq:26}
  (D_{{\bf A}_n}-z)S_n(z)f=f- T_n(z)f\,,\quad f\in L^2(\R^2;\C^2)\,. 
\end{equation*}
This yields, for $n$ large enough, the operator identity
\begin{align*}
  S_n(z)= (D_{{\bf A}_n}-z)^{-1}\big(1- T_n(z) \big)\,,
\end{align*}
from which follows \eqref{eq:21}.

We observe that  \eqref{eq:28}, for $V_n=0$, is a consequence of
\eqref{eq:24},  \eqref{eq:24-b}, and
\begin{align*}
  \Big\| e^{F}(D_{{\bf A}_n}-z)^{-1} e^{-F} \Big\|_{q,r}\le
  \Big\| e^{F}S_n (z)e^{-F} \Big\|_{q,r}\cdot \sum_{k=0}^\infty
  \Big( \big\|e^F T_n(z)e^{-F} \big\|_{q,q}\Big)^k\,.
\end{align*}
Note that the last sum above converges for $n$ large enough due to \eqref{eq:24}.

In order to show \eqref{eq:28} for $V_n\not=0$  we  note that by  H\"older's inequality
\begin{equation}\label{eqeq}
\begin{split}
\|V_n e^F (D_{{\bf A}_n}-z)^{-1}e^{-F}\|_{q,q}\le&
\|V_n\|_\infty\,\|e^F (D_{{\bf A}_n}-z)^{-1}e^{-F}\|_{q,q}\to 0\,,
\end{split}
\end{equation}
 as $n\to\infty$. Therefore, the following computation is meaningful for $n$ large enough
\begin{align*} 
\|e^F&(D_{\mathbf{A}_n}+ V_n-z)^{-1}e^{-F}\|_{q,r}\\ \nonumber
&=\|e^F(D_{\mathbf{A}_n}-z)^{-1}(1+ V_n (D_{\mathbf{A}_n}-z)^{-1})^{-1}e^{-F}\|_{q,r}\\  \nonumber
&\le\|e^F(D_{\mathbf{A}_n}-z)^{-1}e^{-F}\|_{q,r}\sum_{m=0}^\infty 
\|\{V_n e^F(D_{\mathbf{A}_n}-z)^{-1}e^{-F}\}^m\|_{q,q}\,.
\end{align*}
This finishes the proof.
\end{proof}
%
%
%
\section{Gaussian-localization}\label{Gauss}
In this section we show  Theorem \ref{Maintheorem3} on Gaussian localization of eigenfunctions with energies in the  
discrete spectrum of 
\begin{equation}
H= D_{\mathbf A} + V\,,
\end{equation}
under the assumptions (A1)-(A5) stated in the introduction. We choose the magnetic vector potential to be given by
\begin{equation}\label{DefArad1}
{\mathbf A}(\bx):= r^{-1}A(r)\left (\begin{array}{c}-x_2\\ x_1
\end{array}\right),\quad A(r)= r^{-1} \int_0^r B(s)s\,ds\,.
\end{equation}
If $B\in L_{\rm loc}^q(\R^2,\R)$ it is easy to see, using H\"older's
inequality, that if $q\in(1,2]$ then ${\mathbf A}\in L_{\rm
  loc}^p(\R^2;\R^2)$, for some $p\in(2,\infty)$, and that ${\mathbf
  A}\in L_{\rm loc}^\infty(\R^2;\R^2)$ whenever $q\in(2,\infty]$ .

The proof of Theorem \ref{Maintheorem3}, given in Subsection \ref{agmon}, follows the ideas of \cite{Nakamura1996} consisting in Agmon-type estimates with localizations in space and in the angular momentum variable. Of course, we have to adapt the method of \cite{Nakamura1996} since our Hamiltonian is not bounded form below.

In Subsection \ref{vale4} we transform the operator $H$ to polar
coordinates and we decompose it in the angular momentum variable
$m_j$.  The analyticity condition  (A4) on $V$ permits us to obtain
exponential decay in $|m_j|$ of eigenfunctions of $H$ with eigenvalues
$E\in\sigma_{\rm d}(H)$ (see Lemma \ref{RotAnalytic} in Subsection
\ref{vale}). In order to obtain the Agmon estimates, in Subsection
\ref{agmon}, we square the transformed free Dirac operator $K_0^{(2)}$ (see
\eqref{DefK02} for its definition). The Gaussian decay is essentially
due to a positive term in $(K_0^{(2)})^2$ that goes like $r^2$. This
term is in competition with a term that behaves like $m_j$ when $m_j\ge 0$. The
Gaussian weights in the Agmon estimates are localized in the region
where $m_j \lesssim r^2$. The complementary region, on the other hand, is controlled by the exponential decay in $|m_j|$.
\subsection{Unitary transform}\label{vale4}
In the following we  derive an equivalent representation of $H$.
We denote by $U$ the unitary map that represents $H$ in polar coordinates 
(see e.g.  \cite[Section 7.3.3]{Thaller})
\begin{align*}
UHU^*&=H^{(1)}= K^{(1)}_0+v(r,\theta),\\
 K^{(1)}_0&:=S_\theta\big\{ -\ri \partial_r+i r^{-1} 
\sigma_3 J_3-\ri \sigma_3 A(r)\big\},\\
\end{align*}
acting on $\mathcal{H}^{(1)}:=L^2(\R^+)\otimes L^2(T;\C^2)^2$, where 
\begin{align*}
  J_3&:= -\ri \partial_{\theta}+1/2\sigma_3, \quad
S_\theta:=\begin{pmatrix}
 0& e^{-\ri \theta}\\
e^{\ri \theta}&0
\end{pmatrix}\,.
\end{align*}
Next we identify $L^2(T;\C^2)$ with $\ell^2(\Z)^2$
by means of  the transformation $\mathcal{F}:L^2(T;\C^2)\to \ell^2(\Z)^2$  given by  
\begin{align*}
  \mathcal{F}[f](j):=\frac{1}{\sqrt{2\pi}}\int_0^{2\pi}M_\theta e^{-\ri m_j \theta}f(\theta) \,\rd \theta\,,
\end{align*}
for $f\in L^2(T;d\theta)^2$, where  $m_j = (2j+1)/2,\ j\in \Z$, and 
\begin{align*}
  M_\theta:=\left(\begin{array}{cc}
e^{\ri\theta/2}&0\\
0&\ri e^{-\ri \theta/2}\,
\end{array}
\right)\,.
\end{align*}
Under these transformations we find  the decomposition
\begin{align*}
& L^2(\R^2;\C^2) \cong \mathcal{H}^{(2)}:=\bigoplus_{j\in \Z}L^2(\mathbb{R}^+ ;dr)^2\,
\end{align*}
and the corresponding  operator 
\begin{align*}
& H \cong H^{(2)}=K^{(2)}_0+W\,,
\end{align*}
which is essentially self-adjoint on $\mathcal{D}^{(2)}:=\mathcal{F}U C^\infty_0(\R^2;\C^2)$. For $h\in \mathcal{D}^{(2)}$,  $K^{(2)}_0=\mathcal{F}U D_{\mathbf A} U^* \mathcal{F}^*$ acts as
\begin{equation}\label{DefK02}
(K^{(2)}_0h)(r,j)= \big(-\ri \sigma_2 \partial_r + \sigma_1(-m_j r^{-1}+A(r))\big)h(r,j)\,,
\end{equation}
where we used that $\mathcal{F} S_\theta \mathcal{F}^*=\sigma_2$, $\mathcal{F}S_\theta \sigma_3\mathcal{F}^*=\ri \sigma_1$, and that $\mathcal{F} J_3 \mathcal{F}^*$ is the multiplication operator by $m_j$.
The electric potential 
$W=\mathcal{F}v \mathcal{F}^*$ acts as
\begin{equation}\label{defW}
(Wh)(r,l):= \sum_{j \in \Z}
\hat{v}(r,l-j)
h(r,j),
\end{equation}
where
\begin{equation}\label{FourierKoeff}
\hat{v}(r,n) =\frac{1}{2\pi}\int_0^{2\pi}e^{-\ri n\theta}v(r,\theta) \rd \theta\,,\quad n\in\Z\,.
\end{equation}
Two other quantities play an important role in our analysis, namely $W_1:=\mathcal{F}\partial_r v \mathcal{F}^*$ given by
\begin{equation}
  \label{eq:19}
 (W_1 h)(r,l):= \sum_{j \in \Z}
\partial_r \hat{v}(r,l-j)
h(r,j), 
\end{equation}
and $W_2:=\mathcal{F}\partial_\theta v \mathcal{F}^*$ that acts as 
\begin{align}
  \label{eq:29}
  (W_2h)(r,l):= \sum_{j \in \Z}
\ri(j-l)\hat{v}(r,l-j)
h(r,j)\,.
\end{align}
\subsection{Rotation-analyticity}\label{vale}
For $f\in\mathcal{H}^{(1)}$ and  $a\in \R$ we set
\begin{equation}
\big(U_af \big)(r,\theta):= (e^{\ri J_3 a} f)(r,\theta)
=e^{\ri \sigma_3 a/2}f(r, \theta + a)\,.
\end{equation}
%
%
We call a vector $f\in \mathcal{H}^{(1)}$
rotation-analytic, if and only if the series 
$$
\sum_{n\in\N} \frac{\|J_3^n f\|}{n!} \rho^n\,,\quad \rho>0\,,
$$
has an infinite radius of convergence.
%
%
We start by presenting a lemma that gives us some a priori decay of some
eigenfunctions of $H^{(2)}$ in the angular momentum variable.
\begin{lemma}\label{RotAnalytic}
Assume that (A1)-(A4) hold. Let $\Psi\in
\mathcal{H}^{(2)}$ be an eigenfunction of $H^{(2)}$ to the eigenvalue  $E\in \sigma_d(H^{(2)})$.
Then, for every $\gamma>0$, we have
\begin{align}
\sum_{j\in\Z}& 
\int_0^\infty 
e^{2\gamma |m_j|}|\Psi(r,j)|^2 \rd r<\infty.
\end{align}
\end{lemma}
\begin{proof}The proof is analog to the one given in \cite[Section 3]{Nakamura1996}. We sketch it here for the reader's convenience.
Due to Lemma \ref{family2} (in Appendix~\ref{family}) $\{H^{(1)}(z)\}_{z\in\C}$ defined on $\mathcal{D}(K_0^{(1)})$ through 
\begin{align*}
  H^{(1)}(z)=K_0^{(1)}+\widetilde{v}_z\,,
\end{align*}
is an analytic family of type (A)  (see \cite{ReedSimon1978}). Note that when $a\in\R$ the identity $H^{(1)}(a)=U_a H^{(1)}U_a^*$ holds. Moreover, by 
Lemma~\ref{family1} (in Appendix~\ref{family}) we have that
\begin{align*}
\widetilde{v}_z (K_0^{(1)}-\ri)^{-1}
\end{align*}
is a compact operator in $\mathcal{H}^{(1)}$ for any $z\in \C$.
Therefore,
$\sigma_{\mathrm{ess}}(H^{(1)}(z))=\sigma_{\mathrm{ess}}(K^{(1)}_0)$
by Weyl's theorem. Arguing with analytic perturbation theory and using the fact that the spectrum of $H^{(1)}(a)$ and $H^{(1)}$ is the same for $a$ real (see
e.g. the proof of Theorem XIII.36 in \cite{ReedSimon1978} for a similar argument)  we find
that $E\in \sigma_{\mathrm{d}}(H^{(1)})$ of multiplicity $N\in\N$ is also an eigenvalue of
$H^{(1)}(z)$ of the same multiplicity. 

Let   $P_z$ be the $N$-dimensional $E$-eigenprojection of $H^{(1)}(z)$. 
Since rotation-analytic vectors are dense in $\mathcal{H}^{(1)}$ (see e.g. \cite{Nelson1959}) we find some rotation-analytic vectors $f_1,\dots, f_N$ such that 
$\mathrm{Ran} P_0=\mathrm{Span}\{P_0f_1,\dots, P_0f_N\}\,.$ 
Observing that,
for $a\in\R$ and $j\in\{1,\dots,N\}$, 
$$U_aP_0 f_j=P_a U_af_j\,, $$
we find an analytic continuation of $a\mapsto U_aP_0f_j\in \mathcal{H}^{(1)}$ to the complex plane. In particular,  $e^{\ri J_3 z}P_0f_j$ belongs to $\mathcal{H}^{(1)}$ for any $z\in \C$. 
Let $\Psi_1\in \mathrm{Ran}P_0$ be such that $\mathcal{F}\Psi_1=\Psi$. By the discussion above we get that 
$$\mathcal{F} e^{ J_3 \gamma}\Psi_1 \in \mathcal{H}^{(2)}\,,\quad \gamma\in\R\,.$$
This ends the proof since $(\mathcal{F}e^{ J_3 \gamma}\Psi_1)(r,j)=e^{m_j\gamma} \Psi(r,j)$ and 
\begin{align*}
  \sum_{j\in\Z}& 
\int_0^\infty 
e^{2\gamma |m_j|}|\Psi(r,j)|^2 \rd r \\
&\le\sum_{j\in\Z}
\int_0^\infty 
e^{-2\gamma m_j}|\Psi(r,j)|^2 \rd r + \sum_{j\in\Z}
\int_0^\infty 
e^{2\gamma m_j}|\Psi(r,j)|^2 \rd r <\infty\,.
\end{align*}
\end{proof}
%

%
\subsection{Agmon-type Estimate}\label{agmon}
In this section we deduce the Agmon estimates needed in the proof of Theorem \ref{Maintheorem3}.
They were obtained  in \cite{Nakamura1996} for
the case of a magnetic Schr\"odinger operator. These estimates uses
heavily the fact that the Schr\"odinger  operator is bounded
from below. As we commented before we will obtain these estimates for the square of  the  Dirac operator $K_0^{(2)}$.

Fix a number $\bs>B_0$ and note that, due to  (A2), there exists $R_0>0$ so large that the estimate \eqref{aeqq:1} is fulfilled and  moreover
\begin{equation}\label{r0}
\|\mathbbm{1}_{\{r>R_0\}} B\|<\tilde{B}\,,\quad r>R_0\,.
\end{equation}
We set, for $0<q_2<q_1<1$,
\begin{align}
r(j)&:=\begin{cases}\sqrt{ \frac{4\bs}{(q_1^2-q_2^2)B_0^2}m_j },& m_j\ge 0\\
0,&m_j<0\,,\\
\end{cases}\\
\Omega_{q_1,q_2}
&:=\{ (r,j)\in \R^+\times \Z\,|\,r\ge r(j)\}\,.
\end{align}
Moreover, we define
\begin{align}
\rho(r,j)&:= 
\begin{cases}
q_2 B_0/4(r^2-r(j)^2),& m_j\ge 0,\ r\ge r(j)\\
q_2 B_0 r^2/4& m_j<0\\
0& m_j\ge 0,\ r< r(j)\,.
\end{cases}
\end{align}
Eventually we will  choose $q_2$ to be sufficiently close to $1$.
A direct  calculation shows that
\begin{equation}
|\rho(r,j_1)-\rho(r,j_2)|\le \frac{q_2\bs}{(q_1^2-q_2^2)B_0}|j_1-j_2|.
\end{equation}
Let $\rho_\epsilon:= \rho(1+ \epsilon\rho)^{-1}$.  It is easy to see
that
\begin{equation}\label{Diffrho}
|\rho_\epsilon(r,j_1)-\rho_\epsilon(r,j_2)|\le \frac{q_2\bs}{(q_1^2-q_2^2)B_0}|j_1-j_2|.
\end{equation}
Finally, for $R>R_0$, we fix a smooth function $f_R$ in $r$ with bounded derivatives in $\R^+\times\Z$ satisfying
\begin{equation}
f_R(r,j)=
\begin{cases}
1&   r\ge 2R \textrm{ and } (r,j)\in\Omega_{q_1,q_2}\\
0&  r \le R \textrm{ or } (r,j)\not \in\Omega_{q_1,\mu q_2}\,
\end{cases}
\end{equation}
where $\mu\in (0,1)$ is a fixed number that  will be chosen sufficiently close to $1$. Note that $\Omega_{q_1,q_2}\subset \Omega_{q_1,\mu q_2}$.

%
%
%
%
Let $\Psi$ be the eigenfunction from Theorem~\ref{Maintheorem3} and  $\hat{\Psi}:=\mathcal{F}U \Psi$ be a normalized eigenfunction of $H^{(2)}$ with
corresponding energy $E\in \sigma_d(H^{(2)})$. We set, for $R>R_0$ and $\delta\in(0,\mu)$,
\begin{equation}\label{Defg}
g:=e^{\delta \rho_\epsilon}f_R \hat{\Psi}\,.
\end{equation}
Observe that $\delta$ can be chosen arbitrarily close to $1$.
\begin{lemma}\label{KonjugiertW} 
We
  find constants $R_1>R_0$ and $c>0$ such that, for any $\delta\in(-1,1), \rho>R_1$, and
  $j\in\{0,1,2\}$, 
\begin{align}\label{aeqq:1}
\sup_{\epsilon>0}\|\,\theta_\rho e^{ \delta\,\rho_\epsilon}W_j 
e^{- \delta\,\rho_\epsilon} \theta_\rho \|< c \,,
\end{align}
where $W_0:=W$ and $\theta_\rho:=\mathbbm{1}_{\{r>\rho\}}$.  In particular,
the commutator
\begin{align}
  \label{eq:32}
  [K_0^{(2)},W_0]=-\ri\sigma_2 W_1+\frac{\ri \sigma_1}{r}W_2
\end{align}
satisfies the estimate
\begin{align}
  \label{eq:33}
  \sup_{\epsilon>0}\|\, e^{ \delta\,\rho_\epsilon}f_R[K_0^{(2)},W_0]e^{- \delta\,\rho_\epsilon}\|<2c\,,\quad R>R_1\,.
\end{align}
\end{lemma}
%
%
\begin{proof}
We show \eqref{aeqq:1} only for $j=2$, for the other cases follow analogously.
For any $m\in \Z$ and $r>0$ we define the analytic function $\C\ni z\mapsto h_m(r,z):=e^{-\ri m z} \widetilde{v}_z(r,0)$. 
 Using  \eqref{FourierKoeff} and the decay and  analyticity assumptions on $v$ (A5) we find for any $r>\rho$ (sufficiently large), $m\in\Z$, 
 and  $\gamma\in\R$  that there is a constant $C>0$ such that 
\begin{equation}\label{decV}
\begin{split}
  |\hat{v}(r, m)|=&\frac{1}{2\pi}\big| \int_{0}^{2\pi} h_m(r,\theta) \rd \theta\big|
=\frac{1}{2\pi} \big|\int_{0}^{2\pi} h_m(r,\theta-\ri\gamma) \rd \theta\big|\\
\le& 
\frac{e^{-m\gamma}}{2\pi} \int_{0}^{2\pi} u_{2|\gamma|}(r,\theta)  \rd \theta
\le \|\theta_\rho u_{2|\gamma|}\|_\infty e^{-m\gamma}\le C e^{-m\gamma}\,.
\end{split}
\end{equation}
Here we also used Cauchy's integral theorem and the fact that 
$h_m(r,z)$ is $2\pi$-periodic with respect to ${\rm Re}(z)$. In particular, replacing $\gamma$ by $-\gamma$ in the above estimate we see  that, for $\gamma>0$ and  $m\in\Z$, the bound
$ \theta_\rho|\hat{v}(r, m)|\le C e^{-|m|\gamma}$ holds.   
Therefore, using \eqref{Diffrho}, \eqref{eq:29}, and Young's inequality for $\ell^2(\Z^2;\C^2)$
in combination with the Cauchy-Schwarz inequality for $L^2((0,\infty);\C)$,
we get, for $\gamma$ sufficiently large and every $f\in \mathcal{H}^{(2)}$, that $|\sps{f}{\theta_\rho e^{ \delta\,\rho_\epsilon}W_2 
e^{- \delta\,\rho_\epsilon} \theta_\rho f}|$ is bounded by 
\begin{align*}
\int_0^\infty\sum_{l\in\Z}|f(r,l)|  \sum_{j\in\Z} |\theta_\rho \hat{v}(r,l-j)||l-j| e^{\delta \tau |l-j|} |f(r,j)|\rd r \le \tilde{C} \|f\|^2\,, 
\end{align*}
for some constant $\tilde{C}>0$, where $\tau:=\tfrac{q_2\bs}{(q_1^2-q_2^2)B_0}$. 

Equation \eqref{eq:33} follows from \eqref{aeqq:1} and \eqref{eq:32}. Equation \eqref{eq:32} is a consequence of 
\begin{align}
  \label{eq:34}
  [K_0^{(2)},W]=&\mathcal{F}[K_0^{(1)},v]\mathcal{F}^*\,\\
=&\mathcal{F} (-\ri S_\theta \partial_r v+\frac{S_\theta \sigma_3}{r} \partial_\theta v)\mathcal{F}^*\,,
\end{align}
and the fact that $\mathcal{F} S_\theta \mathcal{F}^*=\sigma_2$ and $\mathcal{F}S_\theta \sigma_3\mathcal{F}^*=\ri \sigma_1$.
\end{proof}
%
Before continuing let us state a simple technical result.
\begin{lemma}\label{wdef}
For any $\gamma\in\R$ we have that $e^{\gamma \rho_\epsilon}f_R \hat{\Psi} \in \mathcal{D}(K_0^{(2)})$.
\end{lemma}
\begin{proof}
Let  $\lambda>0$ and $\eta\in \mathcal{F}U C^\infty_0(\R^2;\C^2)$. First observe that a simple computation shows that
$$(\partial_r e^{\gamma \rho_\epsilon}f_R)e^{-\lambda r}$$ extends to a bounded operator on $\mathcal{H}^{(2)}$.
 In addition,  $e^{\lambda r}\hat{\Psi}\in\mathcal{H}^{(2)}$ by Theorem \ref{Maintheorem}. Therefore, we get by   explicit calculation on $\mathcal{F}U C^\infty_0(\R^2;\C^2)$, that
\begin{align*}
  \sps{K_0^{(2)}\eta}{e^{\gamma \rho_\epsilon}f_R \hat{\Psi}}&=\sps{f_Re^{\gamma \rho_\epsilon}K_0^{(2)}\eta}{\hat{\Psi}}\\
&= \sps{K_0^{(2)}f_Re^{\gamma \rho_\epsilon}\eta}{\hat{\Psi}}+\sps{\ri\sigma_2(\partial_r e^{\gamma \rho_\epsilon}f_R)\eta}{\hat{\Psi}}\\
&= \sps{\eta}{e^{\gamma \rho_\epsilon}f_R K_0^{(2)}\hat{\Psi}}-\ri\sps{\eta}{\sigma_2(\partial_r e^{\gamma \rho_\epsilon}f_R)e^{-\lambda r} \,(e^{\lambda r}\hat{\Psi})}\,.
\end{align*}
Since $\eta$ can be chosen arbitrarily  from the domain of essential self-adjointness of $K_0^{(2)}$ we get the desired result.
\end{proof}
An important role in our analysis is played by the quantity
\begin{equation} \label{defQ}
Q:= \operatorname{Re} \langle K^{(2)}_0e^{\delta\,\rho_\epsilon}g\,|\,K^{(2)}_0e^{-\delta\,\rho_\epsilon}g
 \rangle\,,
\end{equation} 
which is well defined due to Lemma \ref{wdef}.
%
Before we show Theorem~\ref{Maintheorem3} we state two preparatory lemmata whose proofs are given in the   next subsection.
\begin{lemma}\label{LowerboundQ}
There are $R,\epsilon $-independent constants $C_1,C_2>0$ such that, for $R>R_1$ sufficiently large, 
\begin{equation}
Q\ge (C_1R^2-C_2)\|g\|^2\,.
\end{equation}
\end{lemma}
%
%
%
\begin{lemma}\label{UpperboundQ}
There is an $R,\epsilon $-independent constant $C_3$ and an $\epsilon$-independent constant $C(R)$ such that, for $R>R_1$ sufficiently large, 
\begin{equation}
Q\le C_3\|g\|^2+ C(R)\|g\|\,. 
\end{equation}
\end{lemma}
%
%
\begin{proof}[Proof of Theorem \ref{Maintheorem3}]
Fix $\delta,q_1, q_2\in(0,1)$.
Combining Lemma \ref{UpperboundQ} and \ref{LowerboundQ} we find,
for  $R>R_1$ sufficiently large, 
\begin{equation}\label{eq31a}
\|g\|\le (C_1R^2- C_2-C_3)^{-1} C(R) 
\end{equation}
Since the right hand side of \eqref{eq31a} is
independent of $\epsilon$ 
we obtain, by the monotone convergence theorem, 
\begin{equation}
\|e^{\delta\rho}\hat{\Psi}\|^2=\lim_{\epsilon\to 0} \|e^{\delta\rho_\epsilon}\hat{\Psi}\|^2\le ( \sup_{\epsilon>0}\|g\|+\|e^{\delta \rho}(1-f_R)\|)^2 <\infty\,.
\end{equation} 
For  $M>1$ define
\begin{equation}
\widetilde{\Omega}_{q_1,q_2,M}= \Big\{ (r,j)\in \R^+\times \Z\,|\,r^2\ge M r(j)^2\Big\}\,.
\end{equation}
We have that $\widetilde{\Omega}_{q_1,q_2,M}\subset \Omega_{q_1,q_2}$. Thus,  for any $(r,j)\in \widetilde{\Omega}_{q_1,q_2,M}$, we get
\begin{equation}
\rho(r,j)=  \frac{q_2B_0}{4}(r^2-r(j)^2)\ge \frac{q_2B_0}{4} \big(1-\frac{1}{M}\big) r^2\,.
\end{equation}
Therefore, setting $\alpha:= \delta q_2 (1-M^{-1})$, we obtain
\begin{equation}\label{wwww}
\|e^{\alpha B_0/4 r^2} \mathbbm{1}_{\widetilde{\Omega}_{q_1,q_2,M}}\hat{\Psi}\|<\infty\,.
\end{equation}
 If
$(r,j)\not \in\widetilde{\Omega}_{q_1,q_2,M}$ then
\begin{equation}
m_j\ge \frac{(q_1^2-q_2^2)B_0^2r^2}{4M\bs}=:\beta r^2.
\end{equation}
Thus, thanks to Lemma \ref{RotAnalytic} we deduce, for any $\gamma>0$, that
\begin{equation}\label{deca}
\|e^{\beta \gamma r^2} \mathbbm{1}_{\widetilde{\Omega}^c_{q_1,q_2,M}}\hat{\Psi}\| <\infty\,.
\end{equation}
Choosing $\gamma= \alpha/\beta \cdot B_0/4$ and combining \eqref{deca} with \eqref{wwww} we conclude that $\|e^{\alpha B_0/4 r^2} \hat{\Psi}\| <\infty$. The latter holds for $\alpha>0$  arbitrarily close to 1, since $\delta$ and $q_2$
can be chosen arbitrarily close to 1 and $M>1$ can be as large as we want. This proves the theorem.
\end{proof}
%
%
\subsection{Proof of lemmas \ref{LowerboundQ} and  \ref{UpperboundQ} }
Before we give the proof of lemmas  \ref{LowerboundQ}  and  \ref{UpperboundQ} we need a preparatory result.
\begin{lemma}\label{LowerboundK02}
For $R>R_1$ sufficiently large we have that
\begin{equation} 
\|K^{(2)}_0 g\|^2_{\mathcal{H}^{(2)}}
\ge \mu^2 q_2^2B_0^2 \|r\,g\|^2/4-\|r^{-1}\,g\|^2/4
-\bs \,\|g\|^2\,.
\end{equation}
\end{lemma}
%
%
\begin{proof}
Let us write $g=(g^+,g^-)^{\mathrm{T}}$ and $g_j^\pm:= g^\pm(\cdot,j)$.
By Equation \eqref{DefK02} we have
\begin{align*}
\|K^{(2)}_0 g\|^2_{\mathcal{H}^{(2)}}&
=\sum_{j\in  \Z} \Big(\|(\partial_r - m_jr^{-1}+A(r))g_j^+\|^2\\
&\phantom{=\sum_{j\in  \Z} \Big(\|(\partial_r}
\quad +\|(-\partial_r - m_jr^{-1}+A(r))g_j^-\|^2\Big)\,.
\end{align*}
Furthermore, dropping the term $-\partial_r^2$, we get
\begin{align*}
\|(\pm \partial_r - m_jr^{-1}+A(r))g_j^\pm\|^2
&\ge \big \langle g^\pm_j\,|\,\big((m^2_j\mp m_j)r^{-2}+A(r)^2 \big)g_j^\pm\big \rangle\\
&\quad+ \big \langle g^\pm_j\,|\,\mp \partial_r A(r)-2m_j r^{-1}A(r)\big)g_j^\pm\big \rangle\,.
\end{align*}
Observe that  (A2) implies that
\begin{align}
  \label{eq:7}
  \frac{1}{r^2}\int_{0}^r b(s)s ds={\rm o}(1)\,,\quad\mbox{as}\quad r\to\infty\,.
\end{align}
This can be seen by splitting the integral above in the regions where $b(s)s$ is integrable and the  one where $b$ decays in the $L^\infty$- norm. Hence, given $q_3\in(q_1,1)$ we find, using \eqref{eq:7}, a constant $R_2>R_1$ such that, for all $r>R_2$,
\begin{equation}
  \label{eq:8}
\begin{split}
  &B(r)\ge q_3 B_0\,,\qquad A(r)\ge q_1 B_0r/2\,,\\
&|\partial_r A(r)|\le \bs\,,\qquad
A(r)\le \bs r/2\,.
\end{split}
\end{equation}
Therefore, for all $r>R>R_2$, we get
\begin{align*}
  &\|(\pm \partial_r - m_jr^{-1}+A(r))g_j^\pm\|^2\\&\qquad
\ge \sps{g_j^\pm}{(-r^{-2}/4+q_1^2B_0^2 r^2/4-2m_jr^{-1}A(r) -\bs)g_j^\pm}\,,
\end{align*}
where we also use that $(m^2_j\pm m_j)\ge -1/4$.

Assume that $m_j<0$. Since  $q_1>q_2$ and   $A(r)>0$, for $r>R_2$, we find that
\begin{equation}
\|(\pm \partial_r - m_jr^{-1}+A(r))g_j^\pm\|^2 
\ge \big \langle g^\pm_j\,|\,\big(q_2^2B_0^2 r^2/4-r^{-2}/4-\bs \big)g_j^\pm\big \rangle\,.
\end{equation}
Assume now that $m_j\ge 0$.  Recall that $A(r)\le \bs r/2$, for $r>R_2$. Using  that $m_j\le r^2(q_1^2-\mu^2q_2^2)B_0^2/(4\bs)$ on $\operatorname{supp}g \subset \Omega_{q_1,\mu q_2}$ we get
\begin{align*}
\|&(\pm \partial_r - m_jr^{-1}+A(r))g_j^\pm\|^2\\
&\ge \big \langle g^\pm_j\,|\,\big(q_1^2B_0^2 r^2/4-r^{-2}/4-m_j \bs-\bs \big)g_j^\pm\big \rangle\\
&\ge \big \langle g^\pm_j\,|\,\big( \mu^2 q_2^2 B_0^2 r^2/4-r^{-2}/4
-\bs\big)g_j^\pm\big \rangle\,.
\end{align*}
This finishes the proof.
\end{proof}
%
%
\begin{proof}[Proof of Lemma \ref{LowerboundQ}]
Notice that
\begin{align*}
e^{\pm \delta\,\rho_\epsilon}&K^{(2)}_0 e^{\mp \delta\,\rho_\epsilon}
= K^{(2)}_0+ Z^{\pm\rho_\epsilon},\qquad
Z^{\pm\rho_\epsilon}:= \pm \ri \delta\partial_r \rho_\epsilon \sigma_2\,.
\end{align*}
Thus, we have 
\begin{align*}
Q&=\operatorname{Re} \langle (K^{(2)}_0
+Z^{-\rho_\epsilon})g\,|\,(K^{(2)}_0+Z^{\rho_\epsilon})g \rangle\\
&= \| K^{(2)}_0 g \|^2 -\delta^2\|\partial_r\rho_\epsilon g\|^2 
\end{align*}
Since $|\partial_r\rho_\epsilon|\le|\partial_r\rho|\le q_2 B_0 r/2$  we find
\begin{equation}
Q\ge \| K^{(2)}_0 g \|^2-(1/4)\delta^2q_2^2B_0^2\| r g\|^2\,.
\end{equation}
Combining this with Lemma \ref{LowerboundK02} and that $\operatorname{supp}g\subset \{ (r,j)\,|\, r\ge R\}$ we obtain
(recall that $0<\delta<\mu<1$)
\begin{equation}
Q\ge \Big((\mu^2-\delta^2)q_2^2B_0^2 R^2/4-R^{-2}/4-\bs\Big)\|g\|^2\,.
\end{equation}
This concludes the proof.
\end{proof}
%
%
\begin{proof}[Proof of Lemma \ref{UpperboundQ}]
We clearly have
\begin{align}\label{qq}
Q&\le |\sps{K_0^{(2)} e^{\delta \rho_\epsilon}g}{f_R (E-W)\hat{\Psi}}| + 
|\sps{K_0^{(2)} e^{\delta \rho_\epsilon}g}{\sigma_2 (\partial_r f_R) \hat{\Psi}}|\,. 
\end{align}
We analyze each of the above terms separately. Using that $(K_0^{(2)}+W)\hat{\Psi}=E\hat{\Psi}$ and noting that $W f_R$ extends  trivially  to a bounded operator (for $R>R_1$ large enough), we have, for any  $\eta\in\mathcal{F}UC_0^\infty (\R^2;\C^2)$,
\begin{align*}
  \sps{K_0^{(2)} \eta}{f_R (E-W)\hat{\Psi}}&= \sps{(E-W)f_R K_0^{(2)} \eta}{\hat{\Psi}}\\
&=\sps{K_0^{(2)}f_R (E-W)\eta}{\hat{\Psi}}+
\sps{[(E-W)f_R,K_0^{(2)} ]\eta}{\hat{\Psi}}\\
&=\sps{\eta}{(E-W)^2f_R\hat{\Psi}}+\sps{\eta}{[W,K^{(2)}_0]f_R\hat{\Psi}}\\
&\quad+\sps{\eta}{\ri \sigma_2 (\partial_r f_R)(W-E)\hat{\Psi}}\,.
\end{align*}
This identity extends to any $\eta\in \mathcal{D}(K^{(2)}_0)$, in particular, we may choose  $\eta=e^{\delta \rho_\epsilon}g$ (see Lemma~\ref{wdef}). Thus,  using Lemma~\ref{KonjugiertW}, we find  a constant $C>0$, independent of  $R$ and $\epsilon$, such that
\begin{align*}
&|\sps{K_0^{(2)} e^{\delta \rho_\epsilon}g}{f_R (E-W)\hat{\Psi}}|\\
&\qquad\le \|g\|\,
\|e^{\delta \rho_\epsilon} \big[ (E-W)^2f_R+[W,K^{(2)}_0]f_R+\ri \sigma_2 (\partial_r f_R)(W-E)\big]\hat{\Psi}\|\\
&\qquad\le C \|g\| \, \|e^{\delta \rho_\epsilon}\hat{\Psi}\|\le
C\|g\|(\|g\|+\|e^{\delta \rho}(1-f_R)\|)\,.   
\end{align*}
We now treat the second term in \eqref{qq}.
We define the operators $\Upsilon$ and $L$ acting, for any
$h\in\mathcal{H}^{(2)}$ and $ (r,j)\in\R^+\times \Z$, as
\begin{align*}
 & (\Upsilon h)(r,j)=e^{-|m_j|}h(r,j)\,,\\
&(Lh)(r,j)= (2\sigma_1\sigma_2(m_j r^{-1}+A(r))(\partial_rf_R\Upsilon h)(r,j)\,.
\end{align*}
Clearly, since $A(r)$ is bounded on the support of $\partial_rf_R$ (for $R>R_1$ large enough; see \eqref{eq:8}) $L$ is an anti-symmetric  bounded operator on $\mathcal{H}^{(2)}$.
With these definitions we have, using again the eigenvalue equation, that for any $\eta\in\mathcal{F}UC_0^\infty (\R^2;\C^2)$
\begin{align*}
\sps{&K_0^{(2)} \eta}{\sigma_2 (\partial_r f_R) \hat{\Psi}}\\ \nonumber
&=\sps{K_0^{(2)}\sigma_2 (\partial_r f_R)\eta}{  \hat{\Psi}}
+ \sps{\eta}{  \mathbbm{1}_{\operatorname{supp}\partial_rf_R}(\ri \partial^2_r f_R\hat{\Psi}- L \Upsilon^{-1}\hat{\Psi} )}\\ \nonumber
&= \sps{\eta}{\mathbbm{1}_{\operatorname{supp}\partial_rf_R}\big(\sigma_2 (\partial_r f_R)(E-W)\hat{\Psi}+  \ri \partial^2_r f_R\hat{\Psi}- L \Upsilon^{-1}\hat{\Psi} \big)}\,.
\end{align*}
Note that $\Upsilon^{-1}\hat{\Psi}\in \mathcal{H}^{(2)}$ by Lemma \ref{RotAnalytic}.
Next, we extend this identity to $\eta\in \mathcal{D}(K_0^{(2)})$ and replace $\eta$ by $e^{\delta \rho_\epsilon}g$. 
Using that $e^{\delta \rho_\epsilon}\mathbbm{1}_{\operatorname{supp}\partial_rf_R}$ is bounded uniformly in $\epsilon>0$, we find  $\epsilon$-independent constants $C(R),C'(R)>0$ such that
\begin{align*}\label{com1}
  & |\sps{K_0^{(2)}e^{\delta \rho_\epsilon}g}{\sigma_2 (\partial_r
    f_R) \hat{\Psi}}| \\
&\quad\le C'(R)\|g\| \,\|e^{\delta
    \rho}\mathbbm{1}_{\operatorname{supp}\partial_rf_R}\| \big(\|
  \Upsilon^{-1} \hat{\Psi}\|+\|\mathbbm{1}_{\operatorname{supp}\partial_rf_R}W\hat{\Psi}\|\big)\\ 
&\quad\le
  C(R)\|g\|\,,
\end{align*}
where in the last inequality we use again Lemma \ref{RotAnalytic}.
Therefore, we obtain from \eqref{qq} and the above bounds that $$Q\le \|g\|\big(C\|g\|+C\|e^{\delta \rho}(1-f_R)\|+C(R)\big)\,,$$
which concludes the proof.
\end{proof}
%
\appendix

\section{Bounds for the Green function of $D_{{\bf A}_0}$}\label{greenf}
Let%
$$\theta(\bx-\bx'):= \frac{B_0|\bx-\bx'|^2}{4}, \quad\eta(\bx,\bx'):=-\frac{B_0}{2}(x_1x_2'-x_2x_1')\,.$$

\begin{lemma}\label{GreenKernel} 
Let $z\in \R\setminus\sigma(D_{{\bf A}_0})$ and let
$G_0(\bx,\bx',z)$, $\bx, \bx'\in\R^2$, be a representation of the Green kernel of
$(D_{{\bf A}_0}-z)^{-1}$ as $2\times 2$-matrix. Then we have that
\begin{equation}
  \label{eq:17}
 \big\|G_0(\bx,\bx';z)\big\|_{\C^2\otimes\C^2}\le e^{-\theta(\bx-\bx')}\omega(\bx-\bx';z)\,,
\end{equation}
for some function $\omega(\cdot;z)\,:\R^2\rightarrow \R^+$
that satisfies
\begin{equation}
  \label{eq:18}
  \sup_{\bx\in\R^2} |\bx|e^{-\varepsilon |\bx|} \omega(\bx;z)<\infty,\ \varepsilon>0\,.
\end{equation}
\end{lemma}
\begin{proof}
Recall that by Proposition \ref{supersymmetry} we have for $E\not=0$ that 
$ \pm E\in\sigma(D_{{\bf A}_0})$ if and only if $E^2\in \sigma(dd^*)\setminus \{0\}=\sigma(d^*d)\setminus \{0\}$, where 
\begin{equation}\label{ddd} 
d^*d=(\mathbf{p}-\mathbf{A}_0)^2-B_0\,,\quad 
dd^*=(\mathbf{p}-\mathbf{A}_0)^2+B_0\,.
\end{equation} 
A simple computation using \eqref{eq:4} yields, 
for any $z\in \R\setminus\sigma(D_{{\bf A}_0})$, 
\begin{equation}\label{ddd1}
\begin{split}
(D_{{\bf A}_0}-z)^{-1}&=(D_{{\bf A}_0}+z) (D_{{\bf A}_0}^2-z^2)^{-1}
\\ 
&
=\left( \begin{array}{cc}
z(d^*d-z^2)^{-1}&d^*(dd^*-z^2)^{-1}\\
d(d^*d-z^2)^{-1}&z(dd^*-z^2)^{-1}
\end{array}
\right)\,.
\end{split}
\end{equation}
It is well-known that the Green function of $(\mathbf{p}-\mathbf{A}_0)^2$  is given by
\begin{align}\label{ddd2}
[(\mathbf{p}-\mathbf{A}_0)^2-\zeta]^{-1}(\bx,\bx')=
(4\pi)^{-1}&\Gamma(\alpha) e^{\ri \eta(\bx,\bx')}
e^{-\theta(\bx-\bx')} U(\alpha,1, 2\theta(\bx-\bx'))\,,
\end{align}
where $U$ is a confluent hypergeometric function and $\alpha=-1/2(\zeta/B_0-1)\notin -\N$, see for instance \cite[Lemma 2.2]{CorneanNenciu1998}.

Combining \eqref{ddd}, \eqref{ddd1}, and \eqref{ddd2} we obtain that the Green kernel of $D_{{\bf A}_0}$ is given by
\begin{equation}
G_0(\bx,\bx';z)=e^{\ri \eta (\bx,\bx')-\theta(\bx-\bx')}\,
\left(\begin{array}{cc}
\Omega_{11}(\bx,\bx';z) &\Omega_{12}(\bx,\bx';z)\\
\overline{\Omega_{12}(\bx,\bx';z)}&\Omega_{22}(\bx,\bx';z)
\end{array}
\right)\,,
\end{equation}
where we define  $\alpha_\pm=-1/2((z^2\pm B_0)/B_0-1)$ and
\begin{align*}
\Omega_{11}(\bx,\bx';z)&:=(4\pi)^{-1}z\Gamma(\alpha_+)\,U(\alpha_+,1,2\theta(\bx-\bx'))\,,\\
\Omega_{12}(\bx,\bx';z)&:=(4\pi)^{-1}B_0\Gamma(\alpha_-+1)\,U(\alpha_-+1,2,2\theta(\bx-\bx'))\{\ri (x_1-x_1')+(x_2-x_2')\}\,,\\
\Omega_{22}(\bx,\bx';z)&:=(4\pi)^{-1}z\Gamma(\alpha_-)\,U(\alpha_-,1,2\theta (\bx-\bx'))\,.
\end{align*}
Here we also used that  $\tfrac{d}{dt} U(\alpha_,1,t)=-\alpha_-U(\alpha_-+1,2,t)$  (see \cite[Eq. 13.4.22]{abramowitzstegun1965}. Since $-\alpha_\pm\not\in \mathbbm{N}_0$, the bounds \eqref{eq:17} and \eqref{eq:18} 
follow now from  the asymptotic formulas for $U$  \cite[Eq. 13.5.2, Eq. 13.5.7, Eq. 13.5.9]{abramowitzstegun1965}.
\end{proof}
%
%
\section{The family $\{H^{(1)}(z)\}_{z\in \C}$}\label{family}
Throughout this section we assume that  (A1)-(A4) are satisfied and use that notation introduced in Section \ref{Gauss}. Our concern is the family of operators $\{H^{(1)}(z)\}_{z\in \C}$ defined a priori on the dense subspace $UC^\infty_0(\R^2,\C^2)$ of $\mathcal{H}^{(1)}$ as 
\begin{align}\label{f1}
H^{(1)}(z):=K^{(1)}_0+\widetilde{v}_z\,,\quad z\in\C\,.
\end{align}
We first state a technical lemma.
\begin{lemma}\label{compsi}
Let $T$ be a (complex-valued) multiplication operator on $L^2(\R^2,\C^2)$ with 
$T\in L^{p}_{\rm loc}(\R^2,\C^2),\, p\in (2,\infty]$  and $\lim_{n\to\infty} \|\mathbbm{1}_{\{|\bx|>n\}} T\|_\infty=0$. Then, $T$ is relative $\sqrt{\bp^2+1}$- compact. 
\end{lemma}
\begin{proof}
For $n\in \N$ write $T=T_1+T_2$ where $T_1$ is supported inside the ball $B_n(0)\subset\R^2$ and $T_2$ on the complement of $B_n(0)$. Then $T_1$ is relative $\sqrt{\bp^2+1}$- compact  \cite[Theorem 4.1]{Simon2005}. Moreover,
\begin{align*}
  \|T(\bp^2+1)^{-1/2}-T_1(\bp^2+1)^{-1/2}\|\le\| T_2\| \to 0\,, 
\end{align*}
as $n\to \infty$, from which follows the claim. 
\end{proof}
\begin{lemma}\label{family1}
For any $z\in \C$ the operator $\widetilde{v}_z (K_0^{(1)}+\ri)^{-1}$ is compact in $\mathcal{H}^{(1)}$. 
\end{lemma}
\begin{proof}
Let $z\in \C$ and $\tau>0$ with $\tau>|z|$. Due to the inequality  $|\widetilde{v}_z|\le u_\tau$ on $\R^+\times T$ and the fact that $u_\tau\in  L^p(\mathbb{R}^+\times T,r\rd r\,\rd \theta)$ (for some $2<p\le\infty$) we see that  $\widetilde{v}_z$ is well defined on the domain of $K_0^{(1)}$. Let $\widetilde{u}_\tau=U^* u_\tau U$. It suffices to show that 
$U^*\widetilde{v}_zU (D_{\mathbf{A}}+\ri)^{-1}$ is compact in $L^2(\R^2;\C^2)$. This is, however, a consequence of Lemma \ref{compsi} and  the discussion at the end of the proof of Lemma \ref{essentialda}.
\end{proof}
\begin{lemma}\label{family2}
$\{H^{(1)}(z)\}_{z\in \C}$ defined in \eqref{f1} extends to an analytic family of type (A) with domain  $\mathcal{D}(H^{(1)}(z))=\mathcal{D}(K_0^{(1)})$.
\end{lemma}
\begin{proof}
Due to Lemma \ref{family1} we know that, for any $z\in \C$, $H^{(1)}(z)$ extends to a closed operator with $\mathcal{D}(H^{(1)}(z))=\mathcal{D}(K_0^{(1)})$. It is enough to show that, for any $\varphi\in \mathcal{D}(K_0^{(1)})$ the mapping $\C\ni z\mapsto H^{(1)}(z)\varphi \in \mathcal{H}^{(1)}$ is analytic.

By the assumption (A5) we have, for any $(r,\theta)\in \R^+\times T$, that the power series $\widetilde{v}_z(r,\theta)=\sum_{n\in\N_0} v^{(n)}(r,\theta) z^n$ with 
\begin{align}\label{f2}
 v^{(n)}(r,\theta)= \frac{1}{2\pi \ri}\oint_{|\zeta|=s} \frac{\widetilde{v}_\zeta(r,\theta)}{\zeta^{n+1}}d\zeta\,,
\end{align}
for some $s>0$, has an infinite convergence radius. In addition, we clearly get from \eqref{f2} that $ |v^{(n)}(r,\theta)|\le u_{2s}(r,\theta)/s^n$ for any $(r,\theta)\in \R^+\times T$. In particular, we find that 
$$\|v^{(n)} \varphi \|\le \frac{1}{s^n}\|u_{2s} \varphi\|\,,\quad \varphi\in  \mathcal{D}(K_0^{(1)})\,.$$
Therefore, for any $|z|<s$,
\begin{align*}
  v_z\varphi=\sum_{n\in\N_0} v^{(n)}z^n \varphi\,,\quad \varphi\in  \mathcal{D}(K_0^{(1)})\,.
\end{align*}
 This concludes the proof  since $s>0$ can be chosen arbitrarily large.
\end{proof}
\bibliographystyle{plain} 

\end{document}